\documentclass[11pt]{article}

\newif\ifauthornote
\authornotefalse
\newif\ifstoc
\stocfalse

\usepackage[margin=1in]{geometry}
\usepackage{amsmath,amssymb,amsthm,mathtools}
\usepackage{enumitem}
\usepackage{mdframed}
\usepackage{microtype}
\usepackage{xspace}
\usepackage{xcolor}
\definecolor{darkgreen}{RGB}{0,100,0}
\usepackage{tikz}
\usetikzlibrary{patterns.meta}
\usepackage[bookmarks=true,pdfstartview=FitH,colorlinks,linkcolor=darkgreen,filecolor=darkgreen,citecolor=darkgreen,urlcolor=darkgreen]{hyperref}

\usepackage[capitalize,noabbrev]{cleveref}

\newtheorem{theorem}{Theorem}
\newtheorem{lemma}[theorem]{Lemma}

\newtheorem{fact}[theorem]{Fact}
\theoremstyle{definition}

\newtheorem{algo}[theorem]{Algorithm}
\crefname{algo}{Algorithm}{Algorithms}
\theoremstyle{remark}
\newtheorem{remark}[theorem]{Remark}

\newcommand{\name}{{\sf Otter}\xspace}
\newcommand{\R}{\mathbb{R}}
\newcommand{\one}{\mathbf{1}}
\newcommand{\Yside}{\mathcal{Y}}
\newcommand{\Xside}{\mathcal{X}}
\newcommand{\Util}{\mathsf{Util}}

\newcommand{\netYSale}{-\Delta_Y}
\newcommand{\netXSale}{-\Delta_X}
\newcommand{\initialYprice}{\sigma_0}
\newcommand{\ystar}[1]{y_{#1}^*}
\newcommand{\Ystar}{Y^*}
\newcommand{\crossing}{M}
\newcommand{\maxwf}{W}
\newcommand{\DYother}{D_Y^{\mathrm{oth}}}
\newcommand{\DXother}{D_X^{\mathrm{oth}}}
\newcommand{\MYother}{M_Y^{\mathrm{oth}}}
\newcommand{\MXother}{M_X^{\mathrm{oth}}}
\newcommand{\SYother}{S_Y^{\mathrm{oth}}}
\newcommand{\SXother}{S_X^{\mathrm{oth}}}
\newcommand{\welfareGain}[1]{\phi_{#1}}
\newcommand{\welfareGainY}{\phi_Y}
\newcommand{\welfareGainX}{\phi_X}
\newcommand{\playerids}{\mathcal{I}}
\newcommand{\reportv}[1]{\ensuremath{\widetilde{v}_{#1}}\xspace}
\newcommand{\reportq}[1]{\ensuremath{\widetilde{q}_{#1}}\xspace}
\newcommand{\reportu}[1]{\ensuremath{\widetilde{u}_{#1}}\xspace}
\newcommand{\reportr}[1]{\ensuremath{\widetilde{r}_{#1}}\xspace}
\newcommand{\Drop}{\textsf{(Drop)}\xspace}
\newcommand{\NPR}{\textsf{(NPR)}\xspace}

\ifauthornote
\newcommand{\elaine}[1]{{\color{red}[elaine: #1]}\xspace}
\else
\newcommand{\elaine}[1]{}
\fi

\title{\name: A Provably MEV-Resilient Automated Market Maker via Surplus Redistribution\thanks{The order of authorship is randomized. The random order was generated using the
American Economic Association's Author Randomization Tool, and can be verified with
confirmation code \texttt{6SqRwY10BsKd} at
\url{https://www.aeaweb.org/journals/policies/random-author-order/search?RandomAuthorsSearch\%5Bsearch\%5D=6SqRwY10BsKd}.}}
\author{
Elaine Shi\thanks{Carnegie Mellon University and Oblivious Labs.
\texttt{elaineshi@cmu.edu}. Supported through a grant from LayerZero via Oblivious Labs.}
\and
Mengqian Zhang\thanks{Carnegie Mellon University.
\texttt{mengqianzhang@cmu.edu}. Supported through a grant from LayerZero via CMU, and a
CyLab Secure Blockchain Initiative grant.}
\and
Hao Chung\thanks{LayerZero Labs. \texttt{hao.chung@layerzerolabs.org}}
\and
Yuhao Li\thanks{Columbia University. \texttt{yuhaoli@cs.columbia.edu}}
}
\date{}

\begin{document}
\begin{titlepage}
\maketitle
\thispagestyle{empty}

\begin{abstract}
Miner extractable value (MEV) in automated market makers allows block builders to profit from transaction ordering and injected trades, imposing costs on users and contributing to builder centralization. We introduce \name (Optimal Truthful Trading with Excess Redistribution), a two-asset batch AMM that achieves provable MEV resilience when the consensus layer provides censorship resilience and block space is uncongested. In particular, \name makes truthful behavior a dominant strategy for both users and builders. Consequently, a builder cannot profit from strategic deviations, including reordering bids or injecting sybil bids. These guarantees continue to hold even when the builder is itself a trader with intrinsic value.
Moreover, we show that our mechanism maximizes social welfare, in a strong sense, within a natural class of mechanisms satisfying the desired game-theoretic properties.

To achieve these guarantees, we introduce a new paradigm called \emph{surplus redistribution}, 
which provably prevents residual surplus from being captured as MEV by redirecting it to the broader community. Specifically, the pool's output tokens need not be distributed entirely among the users in the current batch. Instead, any residual surplus may be transferred, for example, to a smart contract governed by the decentralized community. The accumulated surplus can subsequently be used to benefit community members in ways that preserve the mechanism's game-theoretic guarantees --- for example, by subsidizing traders' transaction fees, rewarding liquidity providers, or returning assets to the pool to reduce price impact and slippage for future traders.

Our approach relies on the underlying consensus layer to provide censorship resilience. We motivate the necessity of this assumption through an impossibility result showing that the desired game-theoretic guarantees become unattainable when the builder is additionally allowed to censor transactions. Thus, our results also provide a mathematically formal demonstration of how consensus-level security guarantees can fundamentally expand what is achievable at the application (i.e., smart-contract) layer. Finally, we establish additional results characterizing the builder fee structures compatible with our desired incentive guarantees.
\end{abstract}
\end{titlepage}

\tableofcontents
\clearpage

\section{Introduction}
\label{sec:intro}

MEV, short for Miner (or Maximum) Extractable Value, refers to the widespread
phenomenon whereby a block builder leverages its unilateral control over the
inclusion and ordering of transactions within a block to extract profit from
this privileged position. MEV is harmful in multiple respects. It not only
imposes additional costs on ordinary users, but also fosters opaque off-chain
markets in which block builders, traders, and specialized searchers compete to
identify and exploit MEV opportunities. These off-chain ecosystems, in turn,
create strong forces toward centralization. Indeed, a recent measurement study
\ifstoc
\else
by Yang et al.~\cite{buildercentral}
\fi
found that the top two block builders
produce more than 85\% of Ethereum blocks, severely undermining the
decentralized vision of blockchain systems.

In this work, we focus on Automated Market Makers (AMMs), one of the most
widely used applications in decentralized finance (DeFi) and a major source of
MEV. Numerous prior works~\cite{highfreq-dex,park2023conceptual,theorymev} have demonstrated how a
strategic block builder can exploit its control over transaction inclusion and
ordering to launch front-running and back-running attacks, often collectively
referred to as {\it sandwich attacks}. Through such attacks, the block builder
can extract essentially risk-free profits at the expense of ordinary users.

In a typical two-asset AMM, users trade one of two assets, denoted $X$ and $Y$,
against a smart contract, which we henceforth also refer to as the {\it pool}.
Initially, the state of the pool is denoted by $(x_0,y_0)$, meaning that the
pool holds $x_0$ units of $X$ and $y_0$ units of $Y$. The exchange rate for
each trade is determined by a predefined pricing curve. For example, the widely
adopted constant-product rule requires
\[
    (x_0-x)(y_0+y)=C
\]
for some positive constant $C$. Thus, starting from state $(x_0,y_0)$, a user
who contributes $y$ units of $Y$ to the pool receives $x$ units of $X$.
Following the trade, the pool transitions to state $(x_0-x,y_0+y)$, 
preserving the invariant that the product of the reserves remains equal to
$C$.

\subsection{Our Results and Contributions}

Our main contribution is a new two-asset AMM mechanism called
\name (short for ``Optimal Truthful Trading with Excess Redistribution'')
that achieves provable MEV resilience, assuming that the underlying consensus
protocol provides censorship resilience and that block space is uncongested
(i.e., there is ample capacity to include all relevant transactions).
Specifically, compared with today's AMMs, our mechanism satisfies several
desirable properties:

\begin{enumerate}[leftmargin=6mm,itemsep=1pt]

\item
{\bf Provable MEV resilience.}
We prove that a block builder (henceforth simply a {\it builder})
cannot benefit from any strategic manipulation within its control,
including misreporting its true type, manipulating transaction ordering,
or creating pseudonymous identities to inject sybil bids.
Since we assume that the underlying consensus protocol provides censorship
resilience, censoring transactions lies outside the builder's strategy space.
Moreover, this guarantee continues to hold even when the builder is itself
a user with an intrinsic demand to trade; we refer to such an entity as a
{\it builder-as-user}.
In particular, the profit-maximizing strategy for a builder or
builder-as-user is to truthfully report its type (if it has an intrinsic
demand to trade) and refrain from injecting sybil bids.
Our mechanism is inherently {\it order-independent}: its outcome
does not depend on the ordering of transactions within a block.
Consequently, the builder gains no advantage from reordering transactions
and may order them arbitrarily without affecting the outcome.

\item
{\bf Truthful reporting and sybil-proofness.}
\ifstoc
\else
We prove that every user's profit-maximizing strategy is to truthfully
report its type, including both its valuation for the asset being purchased
and its budget, where both quantities may be denominated in units of the
other asset. 
\fi
Furthermore, no user can benefit by creating pseudonymous
identities and submitting additional bids through them---a property we
refer to as {\it sybil-proofness}.

This stands in sharp contrast to today's AMMs. Under common MEV attacks,
most notably sandwich attacks~\cite{highfreq-dex,park2023conceptual,theorymev}, users may be forced to trade at
the worst possible price permitted by their reported slippage
tolerance. Consequently, users face an inherently strategic choice when
setting their slippage tolerance: it must be sufficiently permissive for
the transaction to execute, yet sufficiently restrictive to limit their
exposure to MEV extraction.

Today's users may also have an incentive to delay submitting their bids,
hoping to first observe others' bids and then choose their own actions
strategically. By contrast, the incentive-compatibility guarantees of
our mechanism hold {\it ex post}: even after observing all other bids,
each user, builder, or builder-as-user still maximizes its own profit
by behaving truthfully. Thus, observing others' bids before submitting
one's own confers no strategic advantage, eliminating the incentive for
last-minute bid submission.

\item
{\bf Social welfare optimization.}
Since truthful behavior constitutes an equilibrium for all participants,
including users, builders, and builders-as-users, we show that our AMM
maximizes social welfare in a strong sense at equilibrium.
By contrast, today's AMMs generally encourage both users and builders
to behave strategically, and their resulting equilibrium behavior is
poorly understood. Without a well-characterized equilibrium, it is
difficult for a mechanism designer to provide rigorous guarantees about
economic objectives such as social welfare.
\end{enumerate}

Last but not least, our new AMM is intuitive, conceptually simple, and
computationally efficient, making it a promising candidate for practical
implementation and deployment.

\paragraph{Surplus redistribution as a new paradigm.}
To simultaneously achieve the above desiderata, we introduce a new paradigm
called {\it surplus redistribution}.
Most popular AMMs today~\cite{uniswap}, as well as a large body of the
literature~\cite{credible-ex,ammmechdesign,ammtrilemma}, impose two requirements:
(1) all tokens output by the pool must be completely distributed to the users,
and no surplus should remain; 
and
(2) the resulting state of the pool, after each trade or batch of trades,
must land on the prescribed pricing curve (e.g., the constant-product curve).

We depart from this conventional model. While we still require
the end state of the pool to conform to the pricing curve, 
we allow only a portion of the tokens output by the pool to be  
paid to the users. 
Rather than allowing the remaining surplus to be captured as MEV by privileged
participants, our mechanism redistributes it for the benefit of the broader
community.
For example, the surplus may be transferred to a smart contract governed
through a decentralized mechanism and subsequently allocated in various ways
that benefit the community, such as:

\begin{enumerate}[label=\textnormal{(\arabic*)},leftmargin=8mm,itemsep=1pt,topsep=2pt]
\item
as additional rewards for liquidity providers, thereby incentivizing greater
liquidity provision and, in turn, reducing price impact and slippage for users;

\item
back to the pool itself, thereby deepening liquidity and reducing the cost of
future trades;

\item
to subsidize or offset users' transaction fees; and/or

\item
to the underlying layer-one blockchain or its stakeholders, where it can
support the continued development and maintenance of the blockchain,
or provide economic incentives that strengthen the security and robustness
of the broader ecosystem.
\end{enumerate}

We summarize 
our main contribution in the following theorem:

\begin{theorem}[MEV-resilient AMM mechanism]
Suppose that the underlying consensus layer ensures 
censorship resilience, and that we allow surplus redistribution. Then, there exists 
an AMM mechanism
that incentivizes truthful reporting, achieves strategy-proofness
for a block builder who may or may not have intrinsic demand, 
and maximizes social welfare 
in a strong sense (to be formally defined later).
Further, the mechanism can be computed in $O(n \log n)$ operations
for a batch of $n$ orders.
\end{theorem}

\paragraph{Mitigating the externalities of an already centralized builder economy.}
We argue that deploying our new AMM {\it alongside, and in direct competition with},
legacy AMMs can significantly mitigate the negative externalities arising from
today's highly centralized builder economy, in the following ways:

\begin{itemize}[leftmargin=6mm,itemsep=1pt]

\item
{\bf Eliminating builders' privileged control over economic surplus.}
Today, MEV extraction has contributed to a severe centralization of the builder
economy, which has effectively become a duopoly~\cite{buildercentral}.
This concentration allows a small number of dominant builders to capture a
disproportionate share of the economic surplus generated through MEV.
By contrast, our new AMM prevents this surplus from being captured as MEV and
instead redistributes it to the broader decentralized ecosystem.

\item
{\bf Pressuring builders to offer competitive pricing.}
The first three avenues for surplus redistribution described above can all,
directly or indirectly, reduce users' effective trading costs, making our AMM
economically competitive with legacy AMMs.
Under today's AMMs, users targeted by MEV attacks effectively suffer from  
the worst execution price permitted by their reported slippage tolerance.
By contrast, in our new AMM, increases in the execution price arise naturally 
from competition among users, and provably never from MEV extraction.

Deploying our AMM alongside legacy AMMs therefore creates competitive pressure
on the existing builder economy. Builders 
may be compelled to offer users more favorable pricing,
effectively returning a greater share of the economic surplus to users and
liquidity providers rather than capturing it as MEV.

\end{itemize}

\paragraph{Necessity of the censorship-resilience assumption.}
To establish the desirable properties of our AMM, we assume that the
underlying consensus layer provides censorship resilience.
We argue that this assumption is both necessary and reasonable.
We first establish its necessity through the following impossibility result:
if block space is finite and the builder has the ability to censor bids,
then no non-trivial AMM can simultaneously guarantee truthful reporting
for individual users and strategy-proofness for a builder-as-user.

\begin{theorem}[Necessity of censorship resilience]
Suppose that block space is finite and that the block builder can censor bids.
Then, no non-trivial AMM mechanism can simultaneously incentivize truthful
reporting by every individual user and be strategy-proof for a
builder-as-user.
\label{thm:intro-censor-imp}
\end{theorem}

The censorship-resilience assumption is also well motivated, as numerous blockchain projects are actively pursuing censorship resilience at the consensus layer, most notably Ethereum through FOCIL~\cite{focil}. Beyond FOCIL, a variety of other approaches to achieving censorship resilience have been proposed~\cite{mcp,latencyCR}. More broadly, motivated in part by the urgency of mitigating MEV and its negative externalities, many blockchain projects are exploring enhancements to their consensus protocols that provide stronger guarantees over transaction inclusion and ordering. These efforts include censorship resilience~\cite{focil,mcp,latencyCR}, fair sequencing~\cite{decentral-seq,espresso-seq,orderfair00,orderfair01,orderfair02,orderfair03,orderfair04}, and private transaction submission~\cite{encryptmempool00,encryptmempool01,batchdec00,trx,batchdec01,batchdec02,batchdec03,batchdec04}.

Despite this growing effort at the consensus layer, relatively little is
understood about how such guarantees can expand the possibilities for
application-layer mechanism design from a game-theoretic perspective.
In particular, can stronger consensus-layer guarantees enable fundamentally
new mechanism-design feasibility results that would otherwise be impossible?

Our results provide a concrete affirmative answer to this question.
Our AMM construction and
\Cref{thm:intro-censor-imp} jointly establish a separation between consensus layers
with and without censorship resilience: the desired incentive guarantees are
achievable in the former setting, whereas no non-trivial AMM can achieve them
in the latter.
Thus, our results mathematically confirm the community's wide-spread 
intuition that censorship resilience
at the consensus layer can play a fundamental role in mitigating MEV at the
application layer.

\paragraph{Additional results.}
As a by-product of our analysis, we also characterize the permissible forms
of builder compensation when a portion of the surplus is allocated to the
builder. Specifically, we prove the following:

\begin{theorem}[Constraints on surplus returned to the builder]
For any AMM mechanism that incentivizes truthful reporting and is
strategy-proof for a builder-as-user, the amount of surplus returned to the
builder must be constant and independent of the outcome of the AMM.
\label{thm:intro-0-builder-rev}
\end{theorem}

Interestingly, the proofs of both
\Cref{thm:intro-censor-imp} and
\Cref{thm:intro-0-builder-rev}
draw on techniques from the recent literature on transaction fee mechanisms
(TFMs)~\cite{roughgardeneip1559-ec,foundation-tfm}.
Our results thus uncover a useful connection between TFM and AMM mechanism
design, demonstrating how techniques and insights developed in the study of
one can inform the other.

Last but not the least, we also prove that allowing surplus redistribution is essential
for obtaining a ``dream'' AMM mechanism
that simultaneously satisfies our desired incentive compatibility and 
efficiency guarantees. 

\begin{theorem}[The necessity of surplus redistribution]
If we do not allow surplus redistribution, then 
no deterministic AMM mechanism can simultaneously 
incentivize truthful reporting and achieve Pareto optimality for orders
with ``reasonable'' asks (to be formally defined later). 
\label{thm:intro-necessity-burn}
\end{theorem}

\Cref{thm:intro-necessity-burn} also explains 
why recent works on AMM mechanism design
that do not allow surplus redistribution~\cite{ammmechdesign,ammtrilemma}
cannot achieve reasonable notions of social welfare maximization. 

We stress, however, 
that we introduce surplus redistribution not merely as a technical
device to circumvent the impossibility of
\Cref{thm:intro-necessity-burn}, but more importantly,
as a desirable paradigm itself to provably defend against
MEV and substantially reduce the negative externalities
of the already centralized builder economy~\cite{buildercentral}.  

Both 
\Cref{thm:intro-0-builder-rev} and \Cref{thm:intro-necessity-burn}
hold in a strong sense, even when the underlying consensus
offers strong properties including censorship resilience and fair sequencing.

\subsection{Comparison with Existing Approaches of MEV Mitigation} 
\paragraph{Capturing MEV through builder auction.}
A line of work~\cite{drake2023mevburn,domothy2022burningmev,timeboost} has explored capturing MEV through auctions for block-building or transaction-ordering rights. Such auction-based approaches have also been deployed by major blockchain ecosystems, including Ethereum's Proposer-Builder Separation (PBS) framework~\cite{pbs} and Arbitrum's Timeboost~\cite{timeboost}. Broadly, these approaches allocate privileged block-building or transaction-ordering rights through competitive bidding, allowing some of the associated MEV to be captured by the protocol or returned to the broader community. Unfortunately, such approaches may exacerbate builder centralization by favoring participants with superior access to high-value order flow or other MEV opportunities, thereby disadvantaging smaller builders with fewer resources or customers~\cite{buildercentral,burian2024mevcapturedecentralizationexecution}. As the builder economy becomes increasingly concentrated, diminished competition may in turn allow dominant builders to exercise market power, extract monopoly rents, and retain a larger share of the MEV, ultimately undermining the goal of returning MEV to the broader community.

\ifstoc
\else
\paragraph{Application-specific MEV mitigation or recovery at the consensus layer.}
Several academic works, as well as discussions in the blockchain community,
have proposed implementing application-specific MEV mitigation or recovery
directly at the consensus layer.
For example, Ferreira and Parkes~\cite{credible-ex} propose enforcing
application-specific sequencing rules at the consensus layer that restrict
how a block builder may order transactions, thereby limiting its ability to
extract MEV.
In particular, they design sequencing rules for AMM contracts that provide
weak but provable guarantees against MEV extraction.
Beyond enforcing verifiable sequencing rules, 
the cryptocurrency community%
 have also informally
explored the possibility of estimating the amount of MEV extracted by a
builder and requiring the builder to return, or pay as a tax, an amount
proportional to the estimated MEV to the underlying blockchain.

A fundamental drawback of this general approach is that it pushes
application-specific logic into the consensus layer, blurring the
architectural boundary between consensus and applications.
As the ecosystem grows to support an increasingly diverse collection of
applications, incorporating application-specific MEV rules into consensus
could substantially increase the complexity of the consensus layer.
In the extreme, the consensus layer would need to understand increasingly
rich aspects of application semantics in order to determine which builder
behaviors are permissible.

Consensus-level MEV recovery faces an additional challenge: it requires the
consensus layer to identify and quantify MEV across potentially arbitrary
smart contracts.
Such an approach may be inherently difficult to sustain in a permissionless
environment, where developers can continually deploy new or deliberately
obfuscated contracts.
Each new application may require the consensus protocol or its maintainers to
analyze the contract's semantics and devise an application-specific method for
detecting and quantifying its MEV.
This creates an inherently reactive dynamic: as new strategies and contracts
emerge, the consensus layer must continually adapt its MEV-detection
mechanisms to keep pace.

Our approach instead preserves a clean separation of concerns.
We ask the consensus layer to provide only generic, application-agnostic
primitives---most importantly, censorship resilience---while leaving
application-specific incentive design to the application layer.
Our results demonstrate that even such a generic consensus-layer guarantee can
be sufficient to enable strong, provable MEV resilience at the application
layer.

The approach of Ferreira and Parkes~\cite{credible-ex} also faces additional
limitations specific to their model.
Their framework requires trades to be executed sequentially, with the pool
state satisfying the prescribed pricing curve after every individual trade.
Under this restriction, they establish strong impossibility results.
In particular, even arbitrage resilience---a guarantee strictly weaker than
the incentive-compatibility properties we seek---is impossible in their
model.
They therefore consider a substantially weaker MEV-resilience notion:
roughly speaking, if the builder earns risk-free profit, then each user must
receive an execution price at least as favorable as the price it would have
received had its order been the only order in the block.
\fi

\subsection{Additional Related Work}
\label{sec:related} 

\paragraph{Comparison with closely related work.}
Our work is directly inspired by that of Chan, Wu, and Shi~\cite{ammmechdesign}, who construct an AMM that processes trades in batches while ensuring that the pool's end state conforms to the pricing curve. Their AMM achieves the same game-theoretic guarantees as ours, but has two important limitations in comparison. First, their construction relies on a substantially stronger assumption about the consensus layer---namely, fair sequencing. Moreover, the idealized fair-sequencing model assumed in their work is stronger than existing guarantees in the literature~\cite{decentral-seq,espresso-seq,orderfair00,orderfair01,orderfair02,orderfair03,orderfair04}, 
which provide only approximate sequencing fairness. 
It remains unclear whether their idealized notion of fair sequencing can be realized in practice. Second, unlike our AMM, their mechanism does not achieve optimal social welfare.

\paragraph{Automated market makers.}
Milionis, Moallemi, and Roughgarden~\cite{myersonian} study the design of a market maker's demand curve with the dual objectives of maximizing profit and eliciting truthful reports from traders. Their focus is orthogonal to ours, as their model does not seek to capture or mitigate MEV. In particular, traders submit orders directly to the market maker, and their model
does not capture strategic transaction ordering and associated arbitrage opportunities, 
such as front-running or back-running other traders' orders.
\ifstoc
\else
Bartoletti et al.~\cite{maximizemev}
\fi
study how to determine an adversary's optimal MEV-extraction strategy in an AMM. In particular, they develop an algorithmic procedure for strategically ordering user transactions and inserting adversarial trades so as to maximize the value extracted from users. Their work seeks to characterize and optimize MEV extraction, rather than to design mechanisms that prevent or mitigate MEV.

A line of work has explored batch trading, particularly batch clearing at a uniform price~\cite{batchamm00,batchamm01,cow,ramseyer2023augmenting}. 
However, uniform-price batch clearing alone does not guarantee incentive compatibility: when the clearing price depends on orders in the batch, users may benefit from misreporting their valuations or budgets.
For example, Canidio and Fritsch~\cite{batchamm00,batchamm01} propose a uniform-price batch AMM with a modified potential function that achieves arbitrage resilience, 
but only provides incentive compatibility when there is an infinite number of competing arbitrageurs; with finitely many arbitrageurs, strategic misreporting can be profitable. 
Other batch-trading mechanisms~\cite{fairbatchauction} similarly do not satisfy our notion of incentive compatibility.
\ifstoc
\else
Zhang et al.~\cite{batchnotic}
\fi
likewise observe this limitation and study optimal strategic behavior under batch clearing.
A recent work by Li, Shi, and Zhang~\cite{ammtrilemma} establishes a somewhat counterintuitive tension between uniform pricing and incentive compatibility. Although uniform pricing eliminates internal arbitrage, they show that it is fundamentally incompatible with the stronger goal of incentive compatibility: in the absence of surplus redistribution, no efficient AMM mechanism can simultaneously achieve uniform pricing and incentive compatibility.

\paragraph{Other related work.}
\ifstoc
\else
Chitra et al.~\cite{mevredistr}
\fi
propose to redistribute MEV to stake-holders to enhance the blockchain's
economic security; however, their work does not suggest any secure mechanism
for capturing the MEV without leading
to builder centralization. 
In this sense, their work is orthogonal to ours. In particular, the surplus
accumulated through our \name mechanism can also be (in part) redistributed
to stake-holders as they suggest.
\ifstoc
\else
Zhang et al.~\cite{zhang2025rediswap}
\fi
propose a mechanism that captures MEV in AMMs and redistributes it to users and liquidity providers. However, their mechanism does not incentivize
truthful reporting for users. Moreover, it relies on a stronger assumption about the underlying consensus layer: in addition to censorship resilience, the consensus layer must support private transaction submission~\cite{zhang2025rediswap}.

Our surplus-redistribution paradigm is closely related to the distinction between strong and weak budget balance in classical mechanism design. 
Strong budget balance requires that all payments collected from participants be fully redistributed among them, whereas weak budget balance permits a surplus while still ruling out outside subsidies. 
Under standard bilateral trade assumptions, the Myerson–Satterthwaite theorem shows that it is impossible to simultaneously achieve Bayesian incentive compatibility, individual rationality, 
and ex post Pareto optimality, even under weak budget balance~\cite{myerson1983efficient}. 
In our setting, requiring all pool output to be fully distributed to users is a token-wise analog of strong budget balance among the pool and users in the current batch. 
Surplus redistribution relaxes this requirement, paralleling weak budget balance. 
We view this relaxation as a fundamental shift in AMM mechanism design: it enables a provably MEV-resilient AMM 
while making the resulting surplus available for redistribution to the broader blockchain community, with layer-one blockchain and its stakeholders among the possible beneficiaries. 
This addresses the broader question of how the underlying blockchain benefits from the applications it supports~\cite{buterin2025scaling}. 
\Cref{thm:intro-necessity-burn} further establishes that this relaxation is necessary for the desired incentive and efficiency guarantees.

Hartline and Roughgarden~\cite{moneyburnmech} study routing, scheduling, anti-spam, and resource-allocation settings in which conventional monetary transfers are undesirable or technologically unavailable. They design mechanisms that instead require agents to destroy resources --- for example, by expending computation or accepting degraded service --- and use such resource destruction as an instrument for providing incentives. Although our work shares the high-level idea of burning, it differs fundamentally in both setting and objective. In our mechanism, the burnt surplus is redistributed to the broader community and therefore contributes to social welfare; in their framework, by contrast, burnt resources constitute a social loss.
Burning has also been used as a mechanism design tool
in the recent transaction fee mechanism literature~\cite{roughgardeneip1559-ec,foundation-tfm}. 

The elegant work
\ifstoc
\else
of Yokoo et al.~\cite{falsename00}
\fi
considered combinatorial auctions with a finite set of indivisible items. They showed that, under a suitable submodularity condition on the welfare function, the VCG auction is sybil-proof, even though VCG is not sybil-proof in general.
While our approach to achieving Sybil-proofness is reminiscent of that
\ifstoc
\else
of Yokoo et al.~\cite{falsename00}%
\fi, we stress that our auction setting differs substantially from theirs, and thus their proof
does not directly carry over to our setting.  
Specifically, 
in our setting, the tokens being traded are infinitely divisible, trades with the pool are governed by a pricing curve, and agents are subject to budget caps. Moreover, our Sybil strategy space is broader: Sybil identities may not only misreport both valuations and budgets, but may also submit bids in the opposite trading direction.

Our setting resembles multi-unit auctions with budget limits~\cite{adaptiveclinch},
in which buyers compete for a fixed supply of $X$ using $Y$, subject to hard
payment budgets.  Such budgets make utility non-quasilinear and preclude a
standard VCG treatment, motivating their adaptive clinching auction.  At first
glance, it may therefore be surprising that VCG works in our setting despite
the presence of budget constraints.  One might try to reconcile the two by
viewing each budget-constrained buyer as a capacitated seller of $Y$, but this
does not make our procurement-VCG construction applicable to their setting~\cite{adaptiveclinch}:
 because their supply
of $X$ is fixed independently of the amount of $Y$ collected, it corresponds
in our formulation to a constant compensation curve $F$.  The resulting
procurement problem is degenerate---additional input generates no marginal
output, so the welfare-maximizing allocation selects no trade.  In an AMM, by
contrast, additional input generates output along an increasing concave curve,
making seller-side VCG both meaningful and feasible.

\section{Model}
\label{sec:model}
\elaine{state we consider deterministic amm mech earlier? maybe some LB can be for randomized}

Imagine a pool that holds two crypto-assets (also called tokens) named $X$ and $Y$, respectively.
Users can trade with the pool through two types
of orders, called sell-$Y$ and sell-$X$ orders, respectively.

\paragraph{\texorpdfstring{Sell-$Y$ orders.}{Sell-Y orders.}}
A real user $i$ of sell-$Y$ profile has type $(v_i,q_i)\in\R_{\ge0}^2$, indicating
that the user 
has a {\it budget} of $q_i \geq 0$ units of $Y$ to spend on purchasing $X$, 
and its reservation value of each unit of $Y$ is $v_i \geq 0$, measured in units of $X$.   
A sell-$Y$ identity reports a per-unit ask and a declared budget
$(\reportv{i},\reportq{i})\in\R_{\ge0}^2$.  Truthful reporting means
$(\reportv{i},\reportq{i})=(v_i,q_i)$.

\paragraph{\texorpdfstring{Sell-$X$ orders.}{Sell-X orders.}}
Symmetrically, a real user $j$ of sell-$X$ profile has type $(u_j,r_j)\in\R_{\ge0}^2$, where $r_j$ is
its budget in units of $X$, and $u_j$ is its
reservation value for one unit of $X$, measured in units of $Y$.  
A sell-$X$ identity reports
$(\reportu{j},\reportr{j})\in\R_{\ge0}^2$, and truthful reporting means
$(\reportu{j},\reportr{j})=(u_j,r_j)$.

\paragraph{AMM mechanism.}
An AMM mechanism receives as input a list of orders, each of either
sell-$Y$ or sell-$X$ type, and decides an outcome $(x_i, y_i)$ for each order:
\begin{itemize}[leftmargin=6mm,itemsep=1pt,topsep=5pt]
\item 
For a sell-$Y$ order $(\reportv{i}, \reportq{i})$, it means that the user
obtains $x_i$ units in $X$ in exchange for $y_i$
units of $Y$. 
It is required that $0 \leq y_i \leq \reportq{i}$.
\item 
For a sell-$X$ order, it means that the user
obtains $y_i$ units in $Y$ in exchange for $x_i$ 
units of $X$.  
It is required that $0 \leq x_i \leq \reportr{i}$.
\end{itemize}

Unless otherwise noted, we require that an AMM mechanism satisfy the following natural properties: 

\begin{itemize}[leftmargin=6mm,itemsep=1pt]
\item 
{\it Individual rationality.} 
Under truthful reporting, every user obtains nonnegative utility.
\item 
{\it Feasibility.}
Every token paid to an order is funded by tokens supplied by opposite-side
orders or by an AMM trade. No outside subsidy is required. Any surplus compensation
token is {\it burnt} (i.e., redistributed the decentralized community).
\end{itemize}

We assume that the AMM mechanism pays no fee to the builder.
\Cref{sec:builder-impossibility} justifies this assumption by showing that
our desired incentive guarantees preclude outcome-dependent builder fees.
This restriction does not rule out a fixed exogenous payment, funded, for
example, by surplus that the mechanism previously redistributed to the
community.  Provided that this payment is independent of the current
batch's outcome, it has no effect on our
game-theoretic analysis.

\subsection{Strategy Space and Utility}
\label{sec:strategy}

\paragraph{Strategy space.}
A truthful user always reports its true type. 
A strategic user, however, can not only misrepresent its valuation and budget,
but also create an arbitrary finite number of identities,
and inject sybil bids. 
The declared valuation and budget associated with each sybil identity can be
arbitrary finite nonnegative reals.  Moreover, the bids submitted by a
strategic user need not be in the same direction as its true type.
For example, a sell-$Y$ user may submit sell-$X$ bids, and vice versa.

A strategic builder can arbitrarily order the bids in the block,  
and inject sybil bids. 
If the strategic builder also has intrinsic demand, then it
can also misreport its valuation and budget. 
Note that 
the mechanism observes only the declared valuations and budgets,
and it is not aware which identities belong to the same real user or builder-as-user.

\paragraph{Utility.}
Let $\Yside$ and $\Xside$ be the sets of identities that have sell-$Y$ and sell-$X$ profiles,
respectively. 
For the finite set $\playerids$ of identities controlled by one real user, define
aggregate token flows
\begin{equation}
\begin{aligned}
    S_Y(\playerids):=\sum_{i\in\playerids\cap\Yside} y_i, \quad 
    R_X(\playerids)&:=\sum_{i\in\playerids\cap\Yside} x_i, \quad
    S_X(\playerids):=\sum_{i\in\playerids\cap\Xside} x_i, \quad
    R_Y(\playerids) :=\sum_{i\in\playerids\cap\Xside} y_i.
\end{aligned}
\label{eq:portfolio-flows}
\end{equation}
Here $S_T$ denotes the amount of token $T \in \{X, Y\}$ supplied by the user and $R_T$
denotes the amount of token $T$ received by the user.  Its net token changes
are
\begin{equation*}
    \Delta_X(\playerids):=R_X(\playerids)-S_X(\playerids),
    \qquad
    \Delta_Y(\playerids):=R_Y(\playerids)-S_Y(\playerids).
\end{equation*}

A true sell-$Y$
user of type $(v,q)$ has utility
\begin{equation}
    \Util^Y(\playerids;v,q)
    :=
    \begin{cases}
       -\infty,
          &\text{if} \ \netYSale(\playerids)>q\ \text{or}\
           \bigl(\netYSale(\playerids)<0\ \text{and}\
                 \Delta_X(\playerids)<0\bigr) \\
       \Delta_X(\playerids)+v\Delta_Y(\playerids),
          &\text{o.w.}
    \end{cases}
    \label{eq:portfolio-utility-Y}
\end{equation}
Thus a true sell-$Y$ user receives catastrophic utility if its realized net
sale in $Y$ exceeds the true budget $q$, or if its aggregate trade is strictly
in the opposite direction of its desire, namely, if it obtains a strict net
gain in $Y$ while incurring a strict net loss in $X$.
Otherwise its utility is computed in the normal quasilinear way from
its aggregate net trade.  In particular, the utility function does not
penalize a sell-$Y$ user merely for obtaining a net gain in $Y$ when it does
not lose $X$; this includes a potential free-lunch outcome with weak net gains
in both tokens.  Nor does it penalize the user if it gross-sells more than $q$
units of $Y$ through some identities provided opposite-direction identities
buy back enough $Y$ that its final net sale does not exceed $q$.
This makes the possible strategy larger and will later make
our game-theoretic definitions stronger.

Symmetrically, for a true sell-$X$ user of type $(u,r)$,
its utility is defined as:
\begin{equation*}
    \Util^X(\playerids;u,r)
    :=
    \begin{cases}
       -\infty,
          &\text{if} \ \netXSale(\playerids)>r\ \text{or}\
           \bigl(\netXSale(\playerids)<0\ \text{and}\
                 \Delta_Y(\playerids)<0\bigr) \\
       \Delta_Y(\playerids)+u\Delta_X(\playerids),
          &\text{o.w.}
    \end{cases}
\end{equation*}
Thus a true sell-$X$ user receives utility $-\infty$ if it net-sells more
than $r$ units of $X$, or if it obtains a strict net gain in $X$ while
incurring a strict net loss in $Y$.

\subsection{A Concave AMM Pricing Curve}
\label{sec:curve}

Suppose the pool initially holds finite reserves $x_0>0$ and $y_0>0$ of
tokens $X$ and $Y$, respectively.  Let $\widehat y$ denote the curve's
effective $Y$-payout capacity, where $0<\widehat y\le y_0$.
Trades are governed by a pricing curve
\[
    F:(-\widehat y,+\infty)\longrightarrow(-\infty,x_0).
\]
A signed amount $y\in(-\widehat y,+\infty)$ denotes net
$Y$ sent into the AMM, and $F(y)$ denotes net $X$ sent out of the AMM.  Thus
$y>0$ describes selling $Y$ to the AMM and receiving $X$, while $y<0$
describes withdrawing $Y$ while depositing $X$.
We assume:
\begin{enumerate}[label=(A\arabic*),leftmargin=3em,itemsep=1pt,topsep=5pt]
    \item $F(0)=0$;
    \item $F$ is continuous, strictly increasing, and concave;
    \item $F$ is differentiable at $0$ with $0<F'(0)<\infty$;
    \item $F$ has the lower-endpoint behavior $\lim_{p\downarrow-\widehat y}F(p)=-\infty$.
\end{enumerate}

Define the initial prices
\[
    \sigma_0:=F'(0)
    \quad\text{($X$ per $Y$)},
    \qquad
    \rho_0:=\frac{1}{F'(0)}=\frac{1}{\sigma_0}
    \quad\text{($Y$ per $X$)}.
\]

For the sell-$Y$ direction, the mechanism uses the restriction of $F$ to
$[0,+\infty)$.  
In our mechanism description later, whenever we use $F$ 
as a compensation curve for the sell-$Y$ direction, 
we are implicitly using $F$ restricted to $[0,+\infty)$.  
The pool pays
out $F(y)\in[0,x_0)$ units of $X$ upon receiving $y\ge0$ units of $Y$,
and the initial marginal compensation is $\sigma_0$.
For the sell-$X$ direction, 
let
\[
    H:[0,+\infty)\longrightarrow[0,\widehat y),
    \qquad
    H(x):=-F^{-1}(-x).
\]
The pool pays out $H(x)$ units of $Y$ upon receiving $x\ge0$
units of $X$.
Its initial marginal compensation is $\rho_0 = 1/\sigma_0$.
Assumptions (A2) and (A4) ensure that $F^{-1}(-x)$ exists for every $x\ge0$.
Since $F$ is increasing and concave, $H$ is increasing and concave.
Thus the mechanism uses $F|_{[0,+\infty)}$ and $H$ as its two one-sided
compensation curves.  Both are increasing and concave and have unbounded
nonnegative input domains.

\paragraph{Constant product as a special case.}
For a constant-product AMM with initial reserves $(x_0,y_0)$ and
$k=x_0y_0$, take $\widehat y=y_0$.  Then, we have the following: 
\begin{gather*}
    F(y)
    =x_0-\frac{k}{y_0+y},
    \qquad y>-y_0.\\
    F'(y)=\frac{k}{(y_0+y)^2}>0,
    \qquad
    F''(y)=-\frac{2k}{(y_0+y)^3}<0,
    \qquad
    \sigma_0=\frac{x_0}{y_0},
    \qquad
    \rho_0=\frac{y_0}{x_0}.
\end{gather*}

In other words, the pool pays out 
$x_0-\frac{k}{y_0+y}$ units of $X$ for receiving
$y \ge 0$ units of $Y$, 
and it pays out 
  $-F^{-1}(-x) = y_0-\frac{k}{x_0+x}$ units of $Y$ for receiving $x \ge 0$ units of $X$.
Hence constant product is a special case of our generalized concave pricing
curve formulation.

\subsection{Desired Properties}

We would like the AMM mechanism to satisfy
the following desired properties.

\begin{enumerate}[label=\textnormal{(\arabic*)},leftmargin=8mm]

    \item \textbf{User incentive compatibility (a.k.a. truthfulness).}
User incentive compatibility (UIC) requires
that 
truthfully reporting both the reservation value and the budget 
be a dominant strategy over the entire admissible strategy space. In particular, a user cannot benefit from misreporting either its valuation or its budget, 
nor from creating additional identities and injecting sybil bids.

\item \textbf{Strategy proofness for the builder or builder-as-user.}
For either a builder or a builder-as-user, acting honestly is a dominant strategy.
Specifically, the builder or builder-as-user is incentivized
to truthfully report  
its valuation and budget (if there is any intrinsic demand),
refrain from injecting any sybil bids, and order
the bids
in the prescribed manner\footnote{If the AMM mechanism
treats the input bids as a set like ours, 
then the ordering has no effect on the execution outcome.}.

\item
{\bf Pareto optimality for eligible bids.}
A bid is considered {\it eligible} if its ask is no greater than the initial
spot compensation. In other words,
a sell-$Y$ order $(\reportv{}, \reportq{})$ is eligible iff 
$\reportv{} \leq \sigma_0$;
and a sell-$X$ order $(\reportu{}, \reportr{})$ 
is eligible if $\reportu{} \leq \rho_0$. 
We require the following Pareto optimality condition 
for eligible bids: 
under truthful reporting, %
\begin{enumerate}[leftmargin=5mm,itemsep=1pt]
\item 
no subset of users can trade among themselves (subject to available budgets) in a way  
that achieves a Pareto improvement in their joint utilities; and
\item  
no individual user can make an additional trade
direclty with the pool that strictly improves its utility.  
\end{enumerate}
\end{enumerate}

\begin{remark}[Our definitions imply sybil-proofness] 
Since our strategy space allows the user or builder(-as-user)
to create sybil identities and inject sybil bids, 
our user incentive compatibility
and strategy proofness for the builder(-as-user) definitions
directly imply sybil-proofness against these strategic players. 
Sybil-proofness is also referred
to as false-name-proofness in some prior literature~\cite{falsename00,falsename01,falsename02}.
\end{remark}

\begin{remark}[Why restrict Pareto optimality only to eligible bids]
Later in \Cref{sec:ic-le-impossibility}, we prove an impossibility result showing that Pareto optimality cannot be achieved, subject to our incentive-compatibility notions, if non-eligible users are also taken into account. In light of this impossibility, excluding ineligible users from consideration is a natural design choice. Intuitively, the initial spot price can be viewed as setting a baseline price for eligible trades, so that only users whose asks are at least as favorable as this baseline are considered. 
A similar approach has been adopted in prior work on AMM mechanism design~\cite{ammtrilemma}.
\end{remark}

\begin{remark}[Regarding social welfare maximization]
In our formulation, 
the two tokens have no common numeraire, and users may value them at
  different exchange rates. 
It is tricky to define a global notion of social welfare since 
utilities denominated in different tokens cannot be canonically aggregated.
To argue that our mechanism achieves
a reasonable notion of social welfare maximization, we proceed in two steps.
\begin{itemize}[leftmargin=5mm,itemsep=1pt]
\item 
Here we first 
define Pareto optimality for eligible users. Any reasonable notion of welfare
maximization restricted to eligible users should, at a minimum, 
preclude a Pareto improvement for those users. 
\item 
Later in \Cref{sec:sw}, we show that
  our game-theoretic requirements impose strong restrictions 
on the structure of a mechanism. Given these structural restrictions, we
  can subsequently introduce a natural welfare criterion 
for eligible users and prove that our mechanism maximizes it. 
\end{itemize}
Thus, Pareto
  optimality provides a numeraire-free 
efficiency requirement at the model level, while welfare maximization is established later
  within the structure implied by our incentive requirements.
  \end{remark}

\subsection{Basic Facts}
\label{sec:facts}

\paragraph{Relationship between UIC and strategy-proofness for a builder(-as-user).}
We first prove a couple  
facts stating the relationship between UIC
and strategy-proofness for a builder(-as-user). 

\begin{fact}[Strategy-proofness for a builder-as-user implies UIC]
\label{fct:bicimpliesuic}
Strategy-proofness for a builder-as-user implies UIC.
\end{fact}
\begin{proof}
Fix any user's true type and all other users' reports.  If the user had a
profitable deviation, the same user acting as builder could reproduce both
the truthful and deviating executions, including their bid orderings.
The mechanism cannot distinguish builder-owned bids, and pays no builder
fee, so the user's utility in each execution would be unchanged.
The deviation would therefore also benefit the builder-as-user,
contradicting strategy-proofness.
\end{proof}

\begin{fact}[Order-insensitive mechanism: UIC $\Leftrightarrow$ strategy-proofness for 
or builder(-as-user)]
Suppose that the AMM mechanism treats
the input bids as an \emph{unordered} set. 
Then, under our strategy space, 
UIC and strategy proofness
for the builder(-as-user) are equivalent.
\label{fct:uicimpliesbic}
\end{fact}
\begin{proof}
Under our strategy space, 
the only privilege of the builder(-as-user) over an ordinary user  
is its ability to reorder bids within the block.  
However, if the mechanism
is insensitive to ordering, this ability is useless.  
Therefore, if the builder(-as-user)
with intrinsic type $(v, q)$
has a profitable deviation, 
then so does an ordinary user with the same type.
This observation, combined with \Cref{fct:bicimpliesuic}, 
gives the conclusion. 
\end{proof}

Since our proposed mechanism is order-insensitive, \Cref{fct:uicimpliesbic} implies that 
proving UIC is sufficient for establishing 
both UIC and 
strategy-proofness for the builder(-as-user).

\paragraph{Strategy-proofness for a builder(-as-user) implies no arbitrage.}
Some earlier works~\cite{ammmechdesign,ammtrilemma} on AMM mechanism design
also defined an extra property called 
{\it no free lunch}. No free lunch requires
that for any subset $\playerids$ of identities, 
let 
$\Delta_X:=\Delta_X(\playerids)$ and
$\Delta_Y:=\Delta_Y(\playerids)$
be their net gain in the tokens $X$ and $Y$, respectively.
Then, it must be that  
$\Delta_X > 0 \Longrightarrow \Delta_Y < 0$
and 
$\Delta_Y > 0 \Longrightarrow \Delta_X < 0$. 
Intuitively, no free lunch implies
that no user or builder(-as-user) can make risk free profit, 
i.e., the mechanism is {\it arbitrage-free}. 
It also means that for a user or builder with no intrinsic demand and budget,
playing honestly without placing any bids
is a profit-maximizing strategy.
\elaine{this assumes deterministic, state earlier}

The following fact shows that strategy-proofness for a builder(-as-user)
implies the no free lunch property, even when the mechanism depends on
the ordering of bids.
A similar observation was made
\ifstoc
\else
by Chan et al.~\cite{ammmechdesign}%
\fi.

\begin{fact}[Strategy-proofness for a builder(-as-user) implies no free lunch]
\label{fct:uic-implies-nfl}
Suppose that a deterministic AMM mechanism satisfies strategy-proofness
for a builder(-as-user) as defined above.  Suppose also that an honest
zero-budget builder-as-user of either direction receives zero utility.
Then the mechanism satisfies no free lunch for every ordering of the bids.
\end{fact}

\begin{proof}
Fix any ordered list of bids and any subset $\playerids$ of its
identities.  Write $\Delta_X:=\Delta_X(\playerids)$ and
$\Delta_Y:=\Delta_Y(\playerids)$.
Suppose, towards a contradiction, that $\Delta_X>0$ and $\Delta_Y\ge0$;
the other direction is symmetric.  Interpret $\playerids$ as the sybil
identities of a sell-$Y$ builder-as-user with true type $(v,q)=(0,0)$,
and keep all other bids fixed.  The builder can submit these sybil bids
and reproduce the chosen ordering.  Since the mechanism does not observe
bid ownership, this reproduces the same outcome.  The resulting net sale
of $Y$ is $-\Delta_Y\le0=q$, and the builder has a strict net gain in $X$,
so neither catastrophic-utility condition applies.  Under our no-builder-fee
assumption, its utility is therefore
$\Util^Y(\playerids;0,0)=\Delta_X>0$.
Honest behavior gives utility zero by assumption, so this deviation
contradicts strategy-proofness for a builder-as-user.
\end{proof}

For our mechanism, the zero-budget normalization in
\Cref{fct:uic-implies-nfl} is automatic: a truthful report with
$q=0$ (or $r=0$) has zero allocation, and hence zero compensation.  Therefore,
once the full cross-direction UIC guarantee is established in
Section~\ref{sec:sybil}, order-insensitivity and \Cref{fct:uicimpliesbic}
give strategy-proofness for the builder(-as-user), and
\Cref{fct:uic-implies-nfl} yields no free lunch as a corollary.

\section{Warmup: A One-Sided VCG Auction}
\label{sec:generic-seller}
As a warmup, we first present an AMM mechanism
assuming that there are only sell-$Y$ bids
and no sell-$X$ bids are allowed. 
For the opposite case, in which there are only sell-$X$ bids,
a symmetric mechanism can be defined. 
The one-sided construction applies more generally to any compensation curve
\[
    F:[0,+\infty)\longrightarrow\R_{\ge0},
\]
that is continuous, strictly increasing, and concave, with
\[
    F(0)=0,
    \qquad
    \initialYprice:=F'_+(0)\in(0,+\infty).
\]
The restriction of the signed AMM curve from \Cref{sec:curve} to
$[0,+\infty)$ is one such curve.

The rest of this section is organized as follows.
\Cref{sec:vcg-view} describes the auction in VCG terms.
\Cref{sec:efficient-impl} gives an intuitive and efficient implementation
that computes both the allocation and the compensation rules in
quasilinear time, and shows that it is equivalent to the VCG view.
\Cref{sec:one-sided-proofs} proves that the auction satisfies the
desired properties.

\ifstoc
\else
\subsection{One-Sided Auction: VCG View}
\fi
\label{sec:vcg-view}
A reported sell-$Y$ user $i$ has per-unit ask $\reportv{i}\ge0$ and a budget
of $\reportq{i}\ge0$. 
Let $0 \leq y_i \leq \reportq{i}$
be the amount of $Y$ user $i$ 
ends up spending. 
For a finite set $S$ of bids,
define the social welfare
under the outcome $\{y_i\}_{i \in S}$ as
\begin{equation*}
{\sf SW}_F(S, \{y_i\}_{i \in S}) := 
      \underbrace{F\!\left(\sum_{i\in S}y_i\right)}_{\text{$X$ tokens
output by pool}}
       -\underbrace{\sum_{i\in S} \reportv{i}y_i}_{\substack{
\text{cost
of $Y$ tokens}\\[2pt] \text{sent to pool}}},
\end{equation*}

Specifically, the social welfare is defined using $X$
as a numeraire, and   
accounts for the total utility
of the sell-$Y$ users, as well as the burnt amount 
that is effectively redistributed to the community.
In the above definition of ${\sf SW}_F(S)$, the first term is the amount
of $X$ tokens  
output by the pool, and the second term
can be viewed
as a production cost
of the $Y$ tokens sent to the pool.

The maximum achievable social welfare
given a finite set $S$ of sell-$Y$ bids 
is defined as follows: 
\begin{equation}
  \maxwf_F(S)
    :=
    \max_{0\le y_i\le \reportq{i}\ (i\in S)}
    \left\{
{\sf SW}_F(S, \{y_i\}_{i \in S})
    \right\}.
    \label{eq:WG}
\end{equation}

We run a VCG-style auction~\cite{Vickrey61,Clarke71,Groves73} defined as follows. 
\begin{itemize}[leftmargin=6mm,itemsep=1pt]
\item {\bf Allocation rule.}
The allocation rule decides how much $Y$ 
each user spends.
Among welfare maximizers, first choose one maximizing total $Y$ sold 
and then apply an arbitrary exogenous tie-breaking rule.  
Let $\ystar{i}$ be the resulting allocation
chosen for seller $i$, and
let $\Ystar=\sum_i\ystar{i}$.

\item {\bf Compensation rule.}
The compensation rule decides
the amount of $X$ each user obtains as compensation. 
Each seller $i$ gets the 
Clarke pivot compensation defined
as follows:  %
\begin{equation}
\text{let} \ \ x^*_i
    :=
    F(\Ystar)-\sum_{j\ne i}\reportv{j}\ystar{j}-\maxwf_F(S\setminus\{i\}).
    \label{eq:seller-payment}
\end{equation}
Equivalently, let 
\[
    \welfareGain{i}:=\maxwf_F(S)-\maxwf_F(S\setminus\{i\}),
\]
then
user $i$'s compensation can be written as
\begin{equation}
    x^*_i=\reportv{i}\ystar{i}+\welfareGain{i}
    \label{eq:payment-marginal}
\end{equation}
where the first term represents the cost to user $i$
for spending $\ystar{i}$ units of $Y$, 
and the second term represents
a rebate to user $i$ that is equal
to the increase in welfare 
due to user $i$'s existence. 
\end{itemize}

\paragraph{Efficient implementation.}
In general, VCG 
auctions are not necessarily efficient to compute. 
However, in our case, we show a simple and efficient implementation
that takes quasilinear time 
in \Cref{sec:efficient-impl}.

\paragraph{Proof challenges.}
In \Cref{sec:one-sided-proofs}, we prove our one-sided mechanism 
to satisfy UIC, which implies strategy-proofness for the builder(as-user)
due to \Cref{fct:uicimpliesbic}.
Since the mechanism ignores the sequencing
of the orders within the batch, it is easy to see that  
sequencing strategies do not help.
Using standard VCG arguments, we show  
that misreporting value or budget also does not help.  
However, standard VCG arguments do not directly imply that injecting
fake bids is not profitable. To show this, our key idea is to prove that the welfare
function is submodular, and use the properties of submodularity
to make the argument.
Submodularity has also been used to establish Sybil-proofness of VCG
in combinatorial auctions~\cite{falsename00,ausubelmilgrom2006vickrey}.
Those analyses concern a fixed supply of indivisible goods, whereas our
mechanism procures divisible tokens through a concave AMM curve, subject
to privately reported budget caps.  Their results therefore do not apply
directly. %

\subsection{Proofs for the One-Sided Auction}
\label{sec:one-sided-proofs}

We now state and prove the guarantees of the one-sided VCG auction.
By \Cref{thm:efficient-impl}, they apply verbatim to its efficient
implementation in \Cref{alg:one-sided}.
\begin{theorem}[One-sided VCG auction]
\label{thm:generic-seller}
Let $F:[0,+\infty)\to\R_{\ge0}$ be a compensation curve satisfying the
one-sided curve conditions stated at the beginning of this section.
Suppose that only sell-$Y$ bids are allowed.
Then, the one-sided VCG auction of \Cref{sec:vcg-view} satisfies the
following properties:
\begin{enumerate}[label=(\alph*),leftmargin=2.5em]
    \item {\bf Well-formedness}: every user's allocation and compensation are within the desired
    ranges:
    \begin{equation}
        0\le \ystar{i}\le \reportq{i}
        \quad\text{and}\quad
        0\le x_i^*
        \le F(\Ystar)-F(\Ystar-\ystar{i})
        \le \initialYprice\ystar{i};
        \label{eq:individual-payment-upper}
    \end{equation}
    \item {\bf Individual rationality}: every user's truthful utility is
    nonnegative;
    \item {\bf Feasibility}: aggregate VCG compensation is covered by the concave output:
    \begin{equation}
        \sum_i x_i^*\le F(\Ystar);
        \label{eq:aggregate-payment-upper}
\end{equation}
    \item {\bf Submodularity of welfare}: the set function $W_F$ is normalized, monotone, and submodular;
    \item {\bf UIC}: the mechanism satisfies user incentive compatibility (UIC).
\end{enumerate}
\end{theorem}

\begin{proof}
The allocation bounds in part (a) follow directly from the feasible region in
\eqref{eq:WG}.
Under truthful reporting,
\[
    \Util_i=x_i^*-v_i\ystar{i}
       =W_F(S)-W_F(S\setminus\{i\})
       \ge0,
\]
because adding an optional user cannot reduce the optimum.  This proves (b)
and also shows $\welfareGain{i}\ge0$ in \eqref{eq:payment-marginal}.
To obtain the upper bound, delete user $i$ but retain every other user's
selected allocation.  This is feasible for $S\setminus\{i\}$, so
\[
    W_F(S\setminus\{i\})
    \ge
    F(\Ystar-\ystar{i})-\sum_{j\ne i}\reportv{j}\ystar{j}.
\]
Substitution into \eqref{eq:seller-payment} yields
\[
    x_i^*
    \le
    F(\Ystar)-F(\Ystar-\ystar{i}).
\]
Because a concave function has nonincreasing marginal slope and initial right
slope $\initialYprice$,
\[
    F(\Ystar)-F(\Ystar-\ystar{i})\le \initialYprice\ystar{i}.
\]
Nonnegativity follows from \eqref{eq:payment-marginal}.  Together with the
allocation bound proved above, this completes the proof of (a).

Summing the above upper bound on $x_i^*$ gives
\[
    \sum_i x_i^*
    \le
    \sum_i\bigl(F(\Ystar)-F(\Ystar-\ystar{i})\bigr).
\]
If $\Ystar=0$, the claim is immediate since it must be the case that 
for every $i$, 
$\ystar{i} = 0$
and $x^*_i \leq F(\Ystar) - F(\Ystar-\ystar{i}) =0$. 
Otherwise concavity and $F(0)=0$ imply
\[
    F(\Ystar-\ystar{i})
    \ge
    \left(1-\frac{\ystar{i}}{\Ystar}\right)F(\Ystar).
\]
Therefore
\[
    F(\Ystar)-F(\Ystar-\ystar{i})
    \le
    \frac{\ystar{i}}{\Ystar}F(\Ystar).
\]
Summing over $i$ proves \eqref{eq:aggregate-payment-upper}, and hence (c).

We next prove submodularity.  For $t\in[0,\initialYprice]$, let
\[
    D_S(t)
    :=\sum_{i\in S}\reportq{i}\one\{\reportv{i}\le t\}
\]
be aggregate capacity of $Y$ whose ask is at most $t$, and let
\[
    K_F(t)
    :=\sup\{y\ge0:F'_+(y)\ge t\},
\]
where the supremum is allowed to be $+\infty$.
Let $\alpha_S(y)$ be the inverse aggregate supply curve: the minimum reported
per-unit cost of the $y$th infinitesimal unit of the capacity.  For a fixed
total input $Y$, the minimum reported cost is
\[
    \int_0^Y\alpha_S(y)\,dy,
\]
while concavity gives
\[
    F(Y)=\int_0^YF'_+(y)\,dy.
\]
Since $F'_+$ is nonincreasing and $\alpha_S$ is nondecreasing, their difference
is nonincreasing.  Maximizing over $Y$ therefore integrates exactly the region
where marginal compensation is at least marginal cost.  A
layer-cake/Fubini argument gives
\begin{equation}
    W_F(S)
    =
    \int_0^{\initialYprice}
    \min\{D_S(t),K_F(t)\}\,dt.
    \label{eq:seller-layercake}
\end{equation}
For every fixed $t$, the map
\[
    S\longmapsto
    \min\left\{
       \sum_{i\in S}\reportq{i}\one\{\reportv{i}\le t\},
       K_F(t)
    \right\}
\]
is a capped modular set function and is therefore monotone and submodular.
Integration preserves these properties, and $W_F(\varnothing)=0$.  This proves
(d).

Finally, we prove UIC.  Fix all reports not controlled by one real user, and
denote their set by $O$.  Let the user, whose true cost is $vy$ on $[0,q]$,
submit an arbitrary finite collection $\playerids$ of admissible reports --- a 
singleton collection is an ordinary joint misreport of value and cap.  Let
$y_i,x_i$ be the selected sale and compensation of report
$i\in\playerids$, and let
\[
    Y_{\playerids}:=\sum_{i\in\playerids}y_i.
\]
The reports may have budget caps whose sum exceeds $q$.  If $Y_{\playerids}>q$, the real
user's utility is $-\infty$, so the deviation is immediately unprofitable.
It therefore remains only to analyze the case $Y_{\playerids}\le q$.  Put
\[
    \welfareGain{i}
    :=W_F(O\cup\playerids)-W_F(O\cup(\playerids\setminus\{i\})).
\]
By \eqref{eq:payment-marginal},
\[
    x_i=\reportv{i}y_i+\welfareGain{i}.
\]
Submodularity implies
\[
    \sum_{i\in\playerids}\welfareGain{i}
    \le
    W_F(O\cup\playerids)-W_F(O).
\]
Hence the real user's utility from the possibly sybil deviation is
\begin{align*}
    \Util^{\mathrm{Syb}}
    &=\sum_{i\in\playerids}x_i-vY_{\playerids}\\
    &\le
    \sum_{i\in\playerids}\reportv{i}y_i
    +W_F(O\cup\playerids)-W_F(O)-vY_{\playerids}.
\end{align*}
At the chosen allocation for $O\cup\playerids$, the reported costs 
of the sybil identities $\playerids$ cancel
with the same terms inside $W_F(O\cup\playerids)$.  Thus
\[
    \Util^{\mathrm{Syb}}
    \le
    F\!\left(Y_{\playerids}+\sum_{j\in O}y_j\right)
    -\sum_{j\in O}\reportv{j}y_j-vY_{\playerids}-W_F(O).
\]
Because $Y_{\playerids}\le q$, the first three terms are the welfare of a feasible
allocation after merging $\playerids$ into a truthful user of type $(v,q)$
that sells $Y_{\playerids}\le q$.
Therefore
\[
    \Util^{\mathrm{Syb}}
    \le
    W_F(O\cup\{(v,q)\})-W_F(O),
\]
which is exactly the utility from reporting $(v,q)$ truthfully as a single
identity.  Hence no deviation is profitable.  This proves
UIC, including ordinary joint misreports and sybil bids, and establishes (e).
\end{proof}

\subsection{Efficient Implementation}
\label{sec:efficient-impl}

We now describe how to implement our VCG auction efficiently
in quasilinear time. 
We shall focus on the one-sided mechanism of 
\Cref{sec:generic-seller} below. 
Since our full two-sided mechanism in \Cref{sec:mechanism} calls
the one-sided VCG auction as 
a building block, if the one-sided version can be computed 
in quasilinear time, so can the two-sided mechanism. 

Specifically, the VCG view defines the allocation and compensation rules through
optimization problems: the allocation solves \eqref{eq:WG}, and the
compensation \eqref{eq:seller-payment} of every user $i$ additionally
requires the counterfactual optimum $\maxwf_F(S\setminus\{i\})$.
We now give a simple algorithm that computes the allocation and all
compensations in quasilinear time, and we prove that its outcome
coincides exactly with that of the VCG view.
Throughout, let $n:=|S|$.

\begin{figure}[t]
\centering
\begin{tikzpicture}[x=0.68cm,y=1.3cm,>=stealth,font=\footnotesize,
    declare function={rate(\y)=3.5*exp(-\y/5)-0.25;},
    curve/.style={thick,blue!60!black},
    stairs/.style={thick},
    oldstairs/.style={densely dashed,black!50},
    divider/.style={black!35},
    dot/.style={circle,fill,inner sep=0pt,minimum size=4pt},
    ruler/.style={densely dotted,black!60},
    index/.style={font=\scriptsize,black!70},
    welfare/.style={fill=blue!15},
    cost/.style={fill=black!18},
    gain/.style={preaction={fill=orange!12},
      pattern={Lines[angle=45,distance=3pt,line width=0.5pt]},
      pattern color=orange!85!black}]
  \def\Qa{1.4}\def\Qb{4.0}\def\Qc{5.8}\def\Qd{7.4}\def\Qe{8.8}
  \def\ta{8.25}\def\tb{6.7}\def\tc{4.9}\def\td{3.75}\def\te{2.3}
  \pgfmathsetmacro\va{rate(\ta)}\pgfmathsetmacro\vb{rate(\tb)}
  \pgfmathsetmacro\vc{rate(\tc)}\pgfmathsetmacro\vd{rate(\td)}
  \pgfmathsetmacro\ve{rate(\te)}
  \pgfmathsetmacro\Rc{\Qc-\Qb+\Qa}\pgfmathsetmacro\Rd{\Qd-\Qb+\Qa}
  \pgfmathsetmacro\Re{\Qe-\Qb+\Qa}
  \newcommand{\axesandcurve}{%
    \draw[->] (0,0) -- (9.9,0) node[below] {$y$};
    \draw[->] (0,0) -- (0,3.55);
    \draw (0,{rate(0)}) -- (-0.1,{rate(0)}) node[left] {$\initialYprice$};
    \draw[curve] plot[domain=0:9.5,samples=60] (\x,{rate(\x)});
    \node[blue!60!black,above right] at (0.75,{rate(0.75)}) {$F'_+$};}
  \newcommand{\crossingpoints}{%
    \foreach \t/\v/\k/\pos in {\ta/\va/1/above,\tb/\vb/2/above right,
        \tc/\vc/3/above right,\td/\vd/4/above right,\te/\ve/5/above right}
      \node[dot,label={[inner sep=1pt]\pos:$\tau_\k$}] at (\t,\v) {};}
  \begin{scope}
    \fill[cost] (0,0) |- (\Qa,\va) |- (\Qb,\vb) |- (\tc,\vc) -- (\tc,0)
      -- cycle;
    \fill[welfare] (0,\va) -- plot[domain=0:\tc,samples=40] (\x,{rate(\x)})
      -- (\Qb,\vc) -- (\Qb,\vb) -- (\Qa,\vb) -- (\Qa,\va) -- cycle;
    \axesandcurve
    \draw[ruler] (\Qa,\va) -- (\ta,\va);
    \draw[ruler] (\Qb,\vb) -- (\tb,\vb);
    \draw[ruler] (\td,\vd) -- (\Qc,\vd);
    \draw[ruler] (\te,\ve) -- (\Qd,\ve);
    \foreach \q/\v in {\Qa/\va,\Qb/\vb,\Qc/\vc,\Qd/\vd,\Qe/\ve}
      \draw[divider] (\q,0) -- (\q,\v);
    \draw[stairs] (0,\va) -- (\Qa,\va) |- (\Qb,\vb) |- (\Qc,\vc) |- (\Qd,\vd)
      |- (\Qe,\ve);
    \foreach \x/\k in {0.7/1,2.7/2,4.45/3,6.6/4,8.0/5}
      \node[index] at (\x,0.2) {$\k$};
    \foreach \q/\k in {\Qa/1,\Qb/2,\Qc/3,\Qd/4,\Qe/5}
      \draw (\q,0) -- (\q,-0.06) node[below] {$Q_\k$};
    \draw[densely dotted] (\tc,\vc) -- (\tc,-0.06) node[below] {$\Ystar$};
    \crossingpoints
    \node at (1.4,1.45) {$\maxwf_F(S)$};
    \node at (4.9,3.45) {(a) All users};
  \end{scope}
  \begin{scope}[xshift=8.3cm]
    \fill[welfare] (0,\va) -- plot[domain=0:\td,samples=40] (\x,{rate(\x)})
      -- (\Rc,\vd) -- (\Rc,\vc) -- (\Qa,\vc) -- (\Qa,\va) -- cycle;
    \fill[gain] (\Qa,\vb) -- (\Qb,\vb) -- (\Qb,\vc) -- (\tc,\vc)
      -- plot[domain=\tc:\td,samples=20] (\x,{rate(\x)})
      -- (\Rc,\vd) -- (\Rc,\vc) -- (\Qa,\vc) -- cycle;
    \axesandcurve
    \draw[oldstairs] (\Qa,\va) |- (\Qb,\vb) |- (\Qc,\vc) |- (\Qd,\vd)
      |- (\Qe,\ve) -- (\Qe,0);
    \foreach \q/\v in {\Qa/\va,\Rc/\vc,\Rd/\vd,\Re/\ve}
      \draw[divider] (\q,0) -- (\q,\v);
    \draw[stairs] (0,\va) -- (\Qa,\va) |- (\Rc,\vc) |- (\Rd,\vd) |- (\Re,\ve);
    \foreach \x/\k in {0.7/1,2.3/3,4.3/4,5.5/5}
      \node[index] at (\x,0.2) {$\k$};
    \draw[->,black!70] (\Qd,{\vd+0.12}) -- (\Rd,{\vd+0.12});
    \draw[->,black!70] (\Qe,{\ve+0.12}) -- node[above] {slide by $\reportq{2}$}
      (\Re,{\ve+0.12});
    \draw[densely dotted] (\td,\vd) -- (\td,0);
    \crossingpoints
    \node at (1.45,1.5) {$\maxwf_F(S\setminus\{2\})$};
    \node[fill=white,inner sep=1pt] at (2.3,0.865) {$\welfareGain{2}$};
    \node at (4.9,3.45) {(b) Without user $2$};
  \end{scope}
\end{tikzpicture}
\caption{The staircase picture of the one-sided auction with five users,
indexed in nondecreasing order of ask.  The horizontal axis is the
quantity of $Y$ sold, and the vertical axis is measured in units of $X$
per unit of $Y$.  The $k$th step has height $\reportv{k}$ and width
$\reportq{k}$, and the dot labeled $\tau_k$ marks its crossing point,
i.e., the point at which the marginal-rate curve $F'_+$ drops to the
step's height.
(a)~Steps $1$ and $2$ end before their crossing points, whereas steps
$3$, $4$, and $5$ end after theirs (dotted lines).  Hence the auction
stops at step~$3$ and sells $\Ystar=\tau_3$ units in total.  The gray
area is the total reported cost $\sum_j\reportv{j}\ystar{j}$, and the
blue area is the optimal welfare $\maxwf_F(S)$; together, they form the
area $F(\Ystar)$ under the curve.
(b)~Removing user~$2$ deletes its step, and the later steps slide to the
left by $\reportq{2}$ (dashed: the original staircase), whereas the curve
and the crossing points do not change.  The auction now stops at
step~$4$ and sells $\tau_4$ units in total.  The blue area is
$\maxwf_F(S\setminus\{2\})$, and the hatched area between the two
staircases is the welfare gain
$\welfareGain{2}=\maxwf_F(S)-\maxwf_F(S\setminus\{2\})$.}
\label{fig:staircase}
\end{figure}

\paragraph{Intuition: computing the allocation.}
Imagine selling $Y$ to the pool one infinitesimal unit at a time.  When
$y$ units have been sold so far, the next unit $dY$ fetches
$dX=F'_+(y)\,dY$ units of $X$ from the pool; since $F$ is concave, this
marginal rate $\frac{dX}{dY}$ can only decrease as more $Y$ is sold.
The cheapest way to supply the next unit is to take it from a user with
the smallest ask among those with remaining budget, and selling the
unit increases welfare as long as the current marginal rate is at least
this ask.  Hence a welfare-maximizing allocation fills the users'
budgets in nondecreasing order of ask, and stops as soon as the marginal
rate drops below the ask of the user currently being filled; that user
may thus sell only part of its budget, and all subsequent users sell
nothing.

Geometrically, index the users in nondecreasing order of ask, and draw a
staircase whose $k$th step has height $\reportv{k}$ and width
$\reportq{k}$, so that the step spans the quantities $[Q_{k-1},Q_k]$,
where $Q_k$ is the total budget of the first $k$ users.  Draw also the
nonincreasing marginal-rate curve $y\mapsto F'_+(y)$, as in
\Cref{fig:staircase}(a).  The auction sells $Y$ up to the point $\Ystar$
at which the curve drops below the staircase.  The area under the curve
to the left of $\Ystar$ is $F(\Ystar)$, the amount of $X$ output by the
pool, and the area under the staircase to the left of $\Ystar$ is the
total reported cost of the units sold; hence $\maxwf_F(S)$ is the area
between the curve and the staircase to the left of $\Ystar$.  For the
$k$th step, the only relevant information about the curve is a single
number: the \emph{crossing point} $\tau_k$ at which the marginal rate
drops below the step's height $\reportv{k}$, marked by a dot in the
figure.  The $k$th step is used in full if and only if it ends no later
than its crossing point, i.e., $Q_k\le\tau_k$.  As $k$ increases, the
step ends $Q_k$ move to the right, whereas the crossing points $\tau_k$
move to the left, since the nonincreasing curve drops below a higher
step earlier.  Hence the auction stops at the first step $k$ that ends
beyond its crossing point, which can be found by binary search, and the
total quantity sold is $\max\{Q_{k-1},\tau_k\}$ (if there is no such
step, all budgets are exhausted).  In \Cref{fig:staircase}(a), the
auction stops at step~$3$: users $1$ and $2$ exhaust their budgets,
user~$3$ sells $\tau_3-Q_2$ units, and users $4$ and $5$ sell nothing.
We stress that this process determines only how much each user sells:
the compensation of a user is not the amount of $X$ fetched by its own
units, but its Clarke pivot compensation~\eqref{eq:seller-payment}.

\paragraph{Intuition: computing the compensations.}
The compensation of user $i$ depends on $\maxwf_F(S\setminus\{i\})$, the
optimal welfare of the same auction without user~$i$.  Solving each of
these $n$ counterfactual problems from scratch would take quadratic time
overall.  The key observation is that removing user $i$ changes the
staircase in a very structured way: the $i$th step disappears and every
later step slides to the left by exactly $\reportq{i}$, whereas the
heights of all remaining steps, and hence their crossing points, remain
unchanged; see \Cref{fig:staircase}(b), where $i=2$.  Consequently, the
crossing points can be computed once and reused in every counterfactual
problem.  Moreover, every step end of a counterfactual staircase is
obtained in constant time from the original one, by subtracting
$\reportq{i}$ if the step is at or after position $i$, and similarly for
the accumulated costs.  Each counterfactual problem is then solved by a
single binary search for its stopping step, followed by a single
evaluation of $F$.  In \Cref{fig:staircase}(b), for instance, step~$3$
now ends before its crossing point, so the search stops at step~$4$
instead, and only $\tau_4$ units are sold.  Computing the allocation
itself is the special case in which no user is removed.  The figure also
illustrates the welfare gain
$\welfareGain{i}=\maxwf_F(S)-\maxwf_F(S\setminus\{i\})$: since the
staircase without user~$i$ lies on or above the original one,
$\welfareGain{i}$ is the area below the curve and between the two
staircases, and the compensation~\eqref{eq:payment-marginal} of
user~$i$ adds its reported cost $\reportv{i}\ystar{i}$ to this area.

\paragraph{Tie-breaking rule.}
For concreteness, we instantiate the exogenous tie-breaking rule of
\Cref{sec:vcg-view} by a fixed priority order over the identities: among
the welfare maximizers that maximize the total $Y$ sold, the rule
selects the one in which a user sells a positive amount only if every
higher-priority user with the same ask has exhausted its budget.  The
proof of \Cref{thm:efficient-impl} below shows that exactly one such
welfare maximizer exists.  Other tie-breaking rules are discussed in
\Cref{rem:other-tie-breaking}.

\paragraph{Computational model.}
We count arithmetic operations and comparisons on real numbers at unit
cost, and we access the curve $F$ through two primitives: evaluating $F$
at a given point, and computing a \emph{crossing point}
\[
    \tau_F(v,b)
    :=\max\Bigl(\operatorname*{arg\,max}_{0\le y\le b}\,\{F(y)-vy\}\Bigr),
    \qquad v\ge0,\ b\ge0,
\]
i.e., the largest maximizer of $F(y)-vy$ over $y\in[0,b]$, which exists
since $F$ is continuous.  Because $F(y)-vy$ is concave in $y$, it is
nondecreasing on $[0,\tau_F(v,b)]$ and strictly decreasing on
$[\tau_F(v,b),b]$.\footnote{In general, if $g$ is concave on $[0,b]$
and $m$ is its largest maximizer, then $g$ is nondecreasing on $[0,m]$
and strictly decreasing on $[m,b]$.  Indeed, if $y<y'\le m$, then $y'$
lies between $y$ and $m$, so concavity gives
$g(y')\ge\min\{g(y),g(m)\}=g(y)$.  If $m\le y<y'$, then
$y=\lambda m+(1-\lambda)y'$ for some $\lambda\in(0,1]$, so concavity
gives $g(y)\ge\lambda g(m)+(1-\lambda)g(y')>g(y')$, since $g(m)>g(y')$
by the choice of $m$.}
Intuitively, $\tau_F(v,b)$ is the point, clipped to $[0,b]$, at which
the marginal compensation of $F$ drops below $v$.  If $F$ is
differentiable and strictly concave, $\tau_F(v,b)$ is simply the point
at which $F'$ crosses $v$, clipped to $[0,b]$; the definition also
covers curves with kinks or linear pieces.  For example, for the
constant-product curve of \Cref{sec:curve}, $\tau_F(0,b)=b$ and
\[
    \tau_F(v,b)
    =\min\Bigl\{b,\ \max\Bigl\{0,\ \sqrt{x_0y_0/v}-y_0\Bigr\}\Bigr\}
    \qquad\text{for } v>0,
\]
so both primitives take a constant number of arithmetic operations and
square roots.

\paragraph{Preprocessing.}
Sort the users in $S$ in nondecreasing order of ask, breaking ties by
priority, and index them as $1,2,\ldots,n$ in this order, so that
$\reportv{1}\le\reportv{2}\le\cdots\le\reportv{n}$.  For
$k=0,1,\ldots,n$, define the prefix sums
\begin{equation}
    Q_k:=\sum_{j=1}^{k}\reportq{j},
    \qquad
    C_k:=\sum_{j=1}^{k}\reportv{j}\reportq{j}.
    \label{eq:prefix-sums}
\end{equation}
Thus $Q_k$ is the total budget of the first $k$ users, $C_k$ is the
total reported cost of exhausting these budgets, $Q_0=C_0=0$, and $Q_n$
is the total budget of all users.  For $k=1,\ldots,n$, define the
crossing points
\begin{equation}
    \tau_k:=\tau_F(\reportv{k},Q_n).
    \label{eq:crossing-points}
\end{equation}
Each $\tau_k$ depends only on the curve, on the ask $\reportv{k}$, and on
the total budget $Q_n$, which is an upper bound on the total quantity
sold in the actual profile as well as in every counterfactual profile
considered below.  Hence the crossing points are computed only once.

\begin{lemma}[Monotone crossing points]
\label{lem:monotone-crossing}
For every $b\ge0$, the crossing point $\tau_F(v,b)$ is nonincreasing
in~$v$.  In particular, $\tau_1\ge\tau_2\ge\cdots\ge\tau_n$.
\end{lemma}

\begin{proof}
Let $0\le v<v'$, and write $s:=\tau_F(v,b)$ and $s':=\tau_F(v',b)$.
By the optimality of $s$ and $s'$,
\[
    F(s)-vs\ge F(s')-vs',
    \qquad
    F(s')-v's'\ge F(s)-v's.
\]
Adding these two inequalities gives $(v'-v)(s-s')\ge0$, and hence
$s\ge s'$.
\end{proof}

\paragraph{Counterfactual welfare query.}
Once the preprocessing is done, a subroutine  
called $\mathsf{Query}(i)$ outputs a pair given an index $i\in\{0,1,\ldots,n\}$: 
\begin{itemize}[leftmargin=6mm,itemsep=1pt]
\item 
$W^{(-i)}$: 
the optimal welfare of the profile
$S\setminus\{i\}$, in which user $i$ is removed; and 
\item 
$Y^{(-i)}$: the
largest welfare-maximizing total quantity of that profile.  
\end{itemize}
We use $i=0$ as a dummy index that removes nobody; for this index, we
adopt the conventions $\reportv{0}=\reportq{0}:=0$ and
$S\setminus\{0\}:=S$.  Removing user $i$ is equivalent to setting its
budget to zero while keeping its position in the sorted order.
Accordingly, the prefix sums of the profile $S\setminus\{i\}$ are
\begin{equation}
    Q^{(-i)}_k:=Q_k-\reportq{i}\cdot\one\{k\ge i\},
    \qquad
    C^{(-i)}_k:=C_k-\reportv{i}\reportq{i}\cdot\one\{k\ge i\},
    \qquad k=0,1,\ldots,n,
    \label{eq:deleted-prefix-sums}
\end{equation}
each of which takes constant time to compute from~\eqref{eq:prefix-sums};
in particular, no new arrays are needed.  The query finds, by binary
search, the smallest index $k\in\{1,\ldots,n\}$ such that
\begin{equation}
    Q^{(-i)}_k>\tau_k.
    \label{eq:stopping-test}
\end{equation}
Binary search applies because this test is monotone in $k$:
$Q^{(-i)}_k$ is nondecreasing in $k$, whereas $\tau_k$ is nonincreasing
in $k$ by \Cref{lem:monotone-crossing}; hence once
\eqref{eq:stopping-test} holds for some $k$, it holds for all larger
indices.  Finally, $\mathsf{Query}(i)$ returns the pair
$(Y^{(-i)},W^{(-i)})$ defined as follows.
\begin{itemize}[leftmargin=6mm,itemsep=1pt]
\item If no index satisfies \eqref{eq:stopping-test}, all remaining
budgets are exhausted:
\begin{equation}
    Y^{(-i)}:=Q^{(-i)}_n,
    \qquad
    W^{(-i)}:=F\bigl(Y^{(-i)}\bigr)-C^{(-i)}_n.
    \label{eq:query-all-full}
\end{equation}
\item Otherwise, let $k$ be the smallest such index.  The users before
position $k$ exhaust their budgets, user $k$ sells until the total
reaches its crossing point (possibly selling nothing), and the users
after position $k$ sell nothing:
\begin{equation}
    Y^{(-i)}:=\max\bigl\{Q^{(-i)}_{k-1},\,\tau_k\bigr\},
    \qquad
    W^{(-i)}:=F\bigl(Y^{(-i)}\bigr)-C^{(-i)}_{k-1}
        -\reportv{k}\bigl(Y^{(-i)}-Q^{(-i)}_{k-1}\bigr).
    \label{eq:query-stop}
\end{equation}
\end{itemize}
The query uses $O(\log n)$ arithmetic operations and comparisons and a
single evaluation of $F$; in particular, it requires no new
crossing-point computation.

\paragraph{The algorithm.}
Putting these pieces together, we obtain the following algorithm.  The
allocation is obtained from $\mathsf{Query}(0)$, i.e., from the full
profile, by filling the budgets in the sorted order up to the selected
total quantity, and the compensations are obtained from
$\mathsf{Query}(1),\ldots,\mathsf{Query}(n)$.

\begin{mdframed}
\begin{algo}[Efficient one-sided VCG auction]
\label{alg:one-sided}
Given the bids $\{(\reportv{i},\reportq{i})\}_{i\in S}$ and the priority
order:
\begin{enumerate}[label=\arabic*.,leftmargin=6mm,itemsep=1pt]
\item \emph{Sort:} sort the users by nondecreasing ask, breaking ties by
priority, and index them as $1,\ldots,n$ in this order.
\item \emph{Preprocess:} compute the prefix sums $Q_k$ and $C_k$ for
$k=0,\ldots,n$ as in \eqref{eq:prefix-sums}, and the crossing points
$\tau_k$ for $k=1,\ldots,n$ as in \eqref{eq:crossing-points}.
\item \emph{Allocation:} compute $(Y^{(-0)},W^{(-0)}):=\mathsf{Query}(0)$,
which removes nobody.  Let $Y^{\rm alg}:=Y^{(-0)}$ and, for
every $j=1,\ldots,n$, let
\[
    y^{\rm alg}_j:=\min\bigl\{\reportq{j},\ \max\{0,\ Y^{\rm alg}-Q_{j-1}\}\bigr\}.
\]
\item \emph{Compensations:} for every $i=1,\ldots,n$, compute
$(Y^{(-i)},W^{(-i)}):=\mathsf{Query}(i)$, and let
\[
    x^{\rm alg}_i:=\reportv{i}\,y^{\rm alg}_i+W^{(-0)}-W^{(-i)}.
\]
\item \emph{Output:} return $(y^{\rm alg}_i,x^{\rm alg}_i)$ for every
user $i$, under its original identity.
\end{enumerate}
\end{algo}
\end{mdframed}

\paragraph{Equivalence with the VCG view.}
We now show that \Cref{alg:one-sided} computes exactly the outcome of
the VCG view, within the claimed time bound.

\begin{theorem}[Equivalence of the efficient implementation and the VCG view]
\label{thm:efficient-impl}
Let $F$ satisfy the one-sided curve conditions, and let $S$ be a finite
set of $n$ sell-$Y$ bids.  Then \Cref{alg:one-sided} outputs exactly the
outcome of the one-sided VCG auction of \Cref{sec:vcg-view} under the
priority tie-breaking rule, i.e., $y^{\rm alg}_i=\ystar{i}$ and
$x^{\rm alg}_i=x^*_i$ for every user $i\in S$.  Moreover, the algorithm
uses $O(n\log n)$ arithmetic operations and comparisons, $O(n)$
evaluations of $F$, $O(n)$ crossing-point computations, and $O(n)$
space.
\end{theorem}

The core of the proof is the following lemma, which shows that each
query solves its counterfactual problem exactly.

\begin{lemma}[Correctness of the query]
\label{lem:query}
For every $i\in\{0,1,\ldots,n\}$, $\mathsf{Query}(i)$ returns
$W^{(-i)}=\maxwf_F(S\setminus\{i\})$.  Moreover, $Y^{(-i)}$ is the
largest total quantity $\sum_j y_j$ among all optimal solutions
$\{y_j\}$ of \eqref{eq:WG} for the profile $S\setminus\{i\}$.
\end{lemma}

\begin{proof}
Fix $i$, and let $T:=S\setminus\{i\}$.  Since removing user $i$ is
equivalent to setting its budget to zero, we view $T$ as the sorted
users $1,\ldots,n$ with budgets $q'_j:=\reportq{j}$ for $j\neq i$ and
$q'_i:=0$, so that $Q^{(-i)}_k=\sum_{j\le k}q'_j$ and
$C^{(-i)}_k=\sum_{j\le k}\reportv{j}q'_j$.

\emph{Minimum-cost supply.}  For $0\le y\le Q^{(-i)}_n$, let
\[
    \mathcal C(y)
    :=\min\Bigl\{\sum_{j=1}^{n}\reportv{j}y_j:\
        0\le y_j\le q'_j\ (1\le j\le n),\ \sum_{j=1}^{n}y_j=y\Bigr\}
\]
be the minimum reported cost of supplying $y$ units of $Y$.  An exchange
argument shows that $\mathcal C(y)$ is attained by filling the budgets
in the sorted order: if a user sells a positive amount while an earlier
user has unused budget, then shifting a small amount of sale from the
former to the latter does not increase the cost, since the latter's ask
is no larger.  Consequently, for every $j\in\{1,\ldots,n\}$,
\begin{equation}
    \mathcal C(y)=C^{(-i)}_{j-1}+\reportv{j}\bigl(y-Q^{(-i)}_{j-1}\bigr)
    \qquad\text{whenever}\quad Q^{(-i)}_{j-1}\le y\le Q^{(-i)}_j,
    \label{eq:min-cost-interval}
\end{equation}
which holds trivially if this $j$th interval has zero length (e.g., if
$j=i$).  That is, $\mathcal C$ is continuous and piecewise linear, with
slope $\reportv{j}$ on the $j$th interval $[Q^{(-i)}_{j-1},Q^{(-i)}_j]$.

\emph{Reduction to one dimension.}  A feasible allocation for $T$ with
total $y$ has welfare at most $\Psi(y):=F(y)-\mathcal C(y)$, with
equality if and only if its reported cost equals $\mathcal C(y)$.
Hence
\[
    \maxwf_F(T)=\max_{0\le y\le Q^{(-i)}_n}\Psi(y),
\]
and the totals of the optimal solutions of \eqref{eq:WG} for $T$ are
exactly the maximizers of $\Psi$.  It therefore suffices to show that
$Y^{(-i)}$ is the largest maximizer of $\Psi$, and that
$W^{(-i)}=\Psi(Y^{(-i)})$.

By \eqref{eq:min-cost-interval}, on the $j$th interval, $\Psi$
coincides with the concave function $y\mapsto F(y)-\reportv{j}y$ up to
an additive constant.  Recall that $\tau_j$ is the largest maximizer of
this function over $[0,Q_n]\supseteq[0,Q^{(-i)}_n]$; hence the function
is nondecreasing on $[0,\tau_j]$ and strictly decreasing on
$[\tau_j,Q_n]$.

\emph{Case 1: no index satisfies \eqref{eq:stopping-test}.}  Then
$Q^{(-i)}_j\le\tau_j$ for every $j$, so the $j$th interval lies within
$[0,\tau_j]$, and $\Psi$ is nondecreasing on it.  Thus $\Psi$ is
nondecreasing on its entire domain $[0,Q^{(-i)}_n]$.  Its largest
maximizer is therefore $Q^{(-i)}_n=Y^{(-i)}$, and
\eqref{eq:min-cost-interval} gives
$\Psi(Q^{(-i)}_n)=F(Q^{(-i)}_n)-C^{(-i)}_n=W^{(-i)}$.

\emph{Case 2: $k$ is the smallest index satisfying
\eqref{eq:stopping-test}.}  Write $L:=Q^{(-i)}_{k-1}$ and
$z:=Y^{(-i)}=\max\{L,\tau_k\}$.  Since $L\le Q^{(-i)}_k$ and
$\tau_k<Q^{(-i)}_k$, the point $z$ lies in the $k$th interval
$[L,Q^{(-i)}_k]$, so \eqref{eq:min-cost-interval} gives
$\Psi(z)=W^{(-i)}$ as defined in \eqref{eq:query-stop}.  It remains to
show that $z$ is the largest maximizer of $\Psi$.
\begin{itemize}[leftmargin=6mm,itemsep=1pt]
\item \emph{No point to the left of $z$ is better.}  By the minimality
of $k$, every $j<k$ satisfies $Q^{(-i)}_j\le\tau_j$, so $\Psi$ is
nondecreasing on $[0,L]$, as in Case~1.  If $z>L$, then $z=\tau_k$, and
$[L,z]$ lies within both the $k$th interval and $[0,\tau_k]$, so $\Psi$
is nondecreasing on $[L,z]$ as well.  Hence $\Psi(y)\le\Psi(z)$ for all
$y\in[0,z]$.
\item \emph{Every point to the right of $z$ is strictly worse.}  Let
$z<y\le Q^{(-i)}_n$.  Since $z\ge L$, the interval $[z,y]$ is covered
by the intervals $j\ge k$, on which the slope of $\mathcal C$ is
$\reportv{j}\ge\reportv{k}$.  Hence
$\mathcal C(y)-\mathcal C(z)\ge\reportv{k}(y-z)$, and
\[
    \Psi(y)-\Psi(z)
    \le\bigl(F(y)-\reportv{k}y\bigr)-\bigl(F(z)-\reportv{k}z\bigr)
    <0,
\]
where the strict inequality holds because $\tau_k\le z<y\le Q_n$ and
$F(y)-\reportv{k}y$ is strictly decreasing on $[\tau_k,Q_n]$.
\end{itemize}
Therefore $z=Y^{(-i)}$ is the largest maximizer of $\Psi$, which
completes the proof.
\end{proof}

\begin{proof}[Proof of \Cref{thm:efficient-impl}]
We first show that the allocation is correct.  By \Cref{lem:query} with
$i=0$, $Y^{\rm alg}$ is the largest total quantity among all welfare
maximizers for $S$, i.e., $Y^{\rm alg}=\Ystar$ is exactly the total
selected by the quantity-maximizing rule.  As shown in the proof of
\Cref{lem:query}, the welfare maximizers with total $\Ystar$ are exactly
the feasible allocations with total $\Ystar$ whose reported cost is
minimum.  Let $\{y_j\}$ be such an allocation that also obeys the
priority rule, and suppose that some user $j'$ sells a positive amount
while some user $j<j'$, i.e., a user earlier in the sorted order, has
unused budget.  If $\reportv{j}<\reportv{j'}$, then shifting a small
amount of sale from $j'$ to $j$ strictly lowers the reported cost,
contradicting cost minimality; if $\reportv{j}=\reportv{j'}$, then $j$
has a higher priority than $j'$, contradicting the priority rule.
Hence the allocation fills the budgets in the sorted order, and since
its total is $\Ystar$, it must be
$y_j=\min\{\reportq{j},\max\{0,\Ystar-Q_{j-1}\}\}=y^{\rm alg}_j$ for
every~$j$.  Conversely, $\{y^{\rm alg}_j\}$ has total $\Ystar$, has
minimum reported cost by the exchange argument in the proof of
\Cref{lem:query}, and obeys the priority rule.  Therefore it is exactly
the allocation selected by the VCG view, i.e., $y^{\rm alg}_j=\ystar{j}$
for every $j$.

Next, by \Cref{lem:query}, $W^{(-0)}=\maxwf_F(S)$ and
$W^{(-i)}=\maxwf_F(S\setminus\{i\})$ for every $i\in S$.  Hence, by
\eqref{eq:payment-marginal},
\[
    x^{\rm alg}_i
    =\reportv{i}\ystar{i}+\maxwf_F(S)-\maxwf_F(S\setminus\{i\})
    =\reportv{i}\ystar{i}+\welfareGain{i}
    =x^*_i.
\]

Finally, we bound the cost.  Step~1 uses $O(n\log n)$ comparisons.
Step~2 uses $O(n)$ arithmetic operations and $n$ crossing-point
computations.  Steps~3 and~4 make $n+1$ calls to $\mathsf{Query}$, each
using $O(\log n)$ arithmetic operations and comparisons and one
evaluation of $F$, plus $O(n)$ further arithmetic operations.  Step~5 takes $O(n)$ operations.
The algorithm stores $O(n)$ numbers, namely the sorted bids, the prefix
sums, and the crossing points, and the queries can be run one at a time
with $O(1)$ additional space.
\end{proof}

\begin{remark}[Other tie-breaking rules]
\label{rem:other-tie-breaking}
The queries compute the welfare values $\maxwf_F(S)$ and
$\maxwf_F(S\setminus\{i\})$, which do not depend on how ties are broken.
By the exchange argument in the proof of \Cref{thm:efficient-impl}, any
other exogenous tie-breaking rule can only change how users with one and
the same ask split their joint sale.  Given the resulting allocation,
the compensations are still obtained from \eqref{eq:payment-marginal}
with $O(n)$ additional arithmetic operations.
\end{remark}

\section{The Full Two-Sided Mechanism}
\label{sec:mechanism}

We now describe our full \name mechanism,
which extends the one-sided mechanism in \Cref{sec:generic-seller}.

\paragraph{Intuition and additional proof challenges.}
The full mechanism first discards all ineligible bids
whose asks are less favorable than the initial price. 
Among the remaining eligible bids, it then computes the spot-eligible supply on
both sides at the initial spot price. Henceforth, let $D_Y$ and $D_X$ denote the
quantities of $Y$ and $X$, respectively, available to be sold at the initial
spot price or better. Without loss of generality,
if $D_Y \geq \rho_0 D_X$, we say that the sell-$Y$ side is {\it dominant};  
the other case is symmetric. 
The mechanism then fully executes all eligible sell-$X$ orders at the initial
spot price. On the sell-$Y$ side, 
we run the one-sided VCG auction 
of \Cref{sec:generic-seller}
using an adjusted curve $\widetilde F_Y^{\crossing}(y)$, 
obtained by extending 
$F(y)$ so that the first $\crossing := \rho_0 D_X$ units 
of $Y$ sold enjoy the initial spot price.
Equivalently, we can view the pool as being augmented with $D_X$ units of $X$ at the initial spot price.

Our two-sided mechanism raises a new proof challenge: strategic players
can now report bids in the opposite direction of their true type, and such a bid
can replace their true bid under their original identity, or be injected
as fake bids under fake names. 
Our two-sided proofs in \Cref{sec:proof-2-side}
introduce new techniques to prove that injecting opposite-direction bids does not help.
Specifically, the more interesting case is when 
the user's true direction is dominant in the deviating outcome; here, opposite-direction bids can change the very curve used to compute the user's VCG payments.
The effects are two-fold: 
the opposite-direction trades incur a net utility cost, but extending the favorable initial spot-price segment
can benefit the user.
We introduce 
a key technical lemma, the curve-shift lemma (\Cref{lem:crossing-shift}),
which, together with the one-sided Sybil bound, shows that the extra benefit
cannot exceed this net utility cost, even when buying tokens back allows
the user's gross sales to exceed its true budget.

\paragraph{Detailed description.}
A sell-$Y$ report is willing to sell at the initial spot price exactly when
$\reportv{i}\le\sigma_0$.  Define total spot-eligible sell-$Y$ quantity, in units of
$Y$, by
\begin{equation*}
    D_Y
    :=\sum_{i\in\Yside}\reportq{i}\one\{\reportv{i}\le\sigma_0\}.
\end{equation*}
Similarly, define total spot-eligible sell-$X$ quantity, in units of $X$, by
\begin{equation*}
    D_X
    :=\sum_{j\in\Xside}\reportr{j}\one\{\reportu{j}\le\rho_0\}.
\end{equation*}
At the spot price, $D_X$ units of $X$ require $\rho_0D_X$ units of $Y$ as
compensation.  We call a profile \emph{sell-$Y$ dominated} if
\begin{equation}
    D_Y\ge \rho_0D_X,
    \label{eq:Ydom}
\end{equation}
and \emph{sell-$X$ dominated} otherwise.  Ties are resolved in favor of the
sell-$Y$ direction.

\paragraph{\texorpdfstring{Sell-$Y$ dominated profiles}{Sell-Y dominated profiles}.}
Suppose \eqref{eq:Ydom} holds and define
\begin{equation*}
    \crossing:=\rho_0D_X.
\end{equation*}
Thus $\crossing$ is the number of units of $Y$ needed to buy all spot-eligible sell-$X$
supply at the initial price.

First, fully execute every sell-$X$ order with $\reportu{j}\le\rho_0$ at the initial
spot price:
\begin{equation}
    x_j=\reportr{j},
    \qquad
    y_j=\rho_0\reportr{j}.
    \label{eq:minority-X}
\end{equation}
Every sell-$X$ identity with $\reportu{j}>\rho_0$ is assigned zero allocation
and zero compensation.  These
minority-side trades supply $D_X$ units of $X$ and require exactly $\crossing$ units of
$Y$ as compensation.

The dominant sell-$Y$ side faces the following effective compensation curve,
measured in units of $X$ available as a function of units of $Y$ sold:
\begin{equation}
    \widetilde F_Y^{\crossing}(y)
    :=
    \begin{cases}
       \sigma_0 y,&0\le y\le \crossing,\\[1mm]
       D_X+F(y-\crossing),&y\ge\crossing.
    \end{cases}
    \label{eq:Ftilde-Y}
\end{equation}
Because $D_X=\sigma_0\crossing$ and $F'_+(0)=\sigma_0$, the two pieces meet
continuously with the same initial marginal slope.  Since $F$ is concave,
$\widetilde F_Y^{\crossing}$ is increasing and concave, and therefore
satisfies the one-sided curve conditions.

Run the one-sided VCG mechanism of
Section~\ref{sec:generic-seller} on all sell-$Y$ reports with compensation
curve $\widetilde F_Y^{\crossing}$.  
The interpretation of \eqref{eq:Ftilde-Y} is intuitive: the first
$\crossing = \rho_0 D_X$ units of dominant-side $Y$ can be compensated using the $D_X$ units of $X$
supplied by the sell-$X$ orders at the initial spot price.  If more
than $\crossing$ units of $Y$ are sold, only the excess $Y$ is sent to the AMM and the
original curve $F$ is used thereafter.

\begin{fact}
Let $Q_Y$ be the aggregate $Y$ 
spent by the sell-$Y$ users.
Then, $Q_Y \geq \crossing$. 
\label{fct:cross-cover-Y}
\end{fact}
\begin{proof}
By definition, the total reported budget 
of sell-$Y$ identities with ask at most $\sigma_0$ is at least $\crossing$.  On
$[0,\crossing]$, the marginal compensation of $\widetilde F_Y^{\crossing}$ is exactly
$\sigma_0$.  If a welfare maximizer sold a total $Q_Y <\crossing$, some unsold
spot-eligible capacity would remain.  Increasing a seller with ask strictly
below $\sigma_0$ would strictly increase welfare; increasing one with ask
exactly $\sigma_0$ would preserve welfare while increasing total sold quantity.
Either alternative contradicts optimality together with our quantity-maximizing
tie rule.  Hence $Q_Y\ge \crossing$.
\end{proof}

Since $\widetilde F_Y^{\crossing}$ satisfies the one-sided curve
conditions, \Cref{alg:one-sided} computes the one-sided VCG auction on
this curve exactly, by \Cref{thm:efficient-impl}.  The algorithm
accesses the curve only through evaluations and crossing points, and
by~\eqref{eq:Ftilde-Y}, each of these takes at most one call to the
corresponding primitive of $F$ plus $O(1)$ arithmetic operations.  For
instance, $\tau_{\widetilde F_Y^{\crossing}}(v,b)$ equals $0$ if
$v>\sigma_0$, equals $b$ if $v\le\sigma_0$ and $b\le\crossing$, and
equals $\crossing+\tau_F(v,b-\crossing)$ otherwise.  Moreover, by
definition, the curve stays fixed when the auction computes the
counterfactual welfare values that determine the compensations, so the
crossing points are still computed only once.  The remaining steps ---
computing $D_Y$ and $D_X$, determining the dominant side, settling the
minority side, and computing the burnt amount --- take $O(n)$
operations and one more evaluation of $F$.  Hence the full mechanism
can be computed within the bounds of \Cref{thm:efficient-impl}, where
$n$ now denotes the total number of reports; for the constant-product
curve of \Cref{sec:curve}, this amounts to $O(n\log n)$ operations.
The sell-$X$ dominated case below is symmetric, with $H$ in place
of~$F$.

\paragraph{\texorpdfstring{Sell-$X$ dominated profiles}{Sell-X dominated profiles}.}
The sell-$X$ dominated case is 
handled symmetrically. We describe the full details below for completeness.
Suppose $D_Y<\rho_0D_X$. 
Define the amount of $X$ needed to buy all spot-eligible sell-$Y$ supply at the
initial spot price:
\begin{equation*}
    \crossing:=\sigma_0D_Y=\frac{D_Y}{\rho_0}.
\end{equation*}

First, fully execute every sell-$Y$ order with $\reportv{i}\le\sigma_0$ at the initial
spot price:
\begin{equation*}
    y_i=\reportq{i},
    \qquad
    x_i=\sigma_0\reportq{i}.
\end{equation*}
Every sell-$Y$ identity with $\reportv{i}>\sigma_0$ is assigned zero allocation
and zero compensation.  These
minority-side trades supply $D_Y$ units of $Y$ and require exactly $\crossing$ units of
$X$ as compensation.

The dominant sell-$X$ side faces the effective compensation curve, now in
units of $Y$ available as a function of units of $X$ sold,
\begin{equation*}
    \widetilde H_X^{\crossing}(x)
    :=
    \begin{cases}
       \rho_0 x,&0\le x\le \crossing,\\[1mm]
       D_Y+H(x-\crossing),&x\ge\crossing.
    \end{cases}
\end{equation*}
Because $D_Y=\rho_0\crossing$ and $H'_+(0)=\rho_0$, this is increasing and concave.
It therefore satisfies the one-sided curve conditions.
Run the one-sided VCG mechanism of \Cref{sec:generic-seller}
on all sell-$X$ reports using
$\widetilde H_X^{\crossing}$.

\begin{fact}
Let $Q_X$ denote the aggregate $X$ spent by sell-$X$ users.  
Then, $Q_X \ge \crossing$. 
\label{fct:cross-cover-X}
\end{fact}
\begin{proof}
The proof is symmetric to that of \Cref{fct:cross-cover-Y}.
\end{proof}

\begin{theorem}[Main theorem: our full two-sided mechanism]
\label{thm:main}
Suppose that the pricing curve $F$ satisfies assumptions (A1)--(A4)
defined in \Cref{sec:curve}. Our full two-sided mechanism is well-formed and
satisfies individual rationality, feasibility, user incentive compatibility
(UIC), strategy-proofness against a builder(-as-user), and Pareto optimality
for eligible users.
\end{theorem}

The remainder of the paper proves the theorem.  In comparison with our
one-sided guarantees in Theorem~\ref{thm:generic-seller}, we must now account for deviations
that can change which side is dominant or submit reports in both directions.

\section{Proofs for the Two-Sided Mechanism}
\label{sec:proof-2-side}

\subsection{Well-Formedness, Individual Rationality, and Feasibility}
\label{sec:feasibility}
Without loss of generality, we prove well-formedness, individual rationality, and 
feasibility only for
the sell-$Y$ dominated case, since the other case is symmetric.

\paragraph{Well-formedness and individual rationality.}
Well-formedness and individual rationality for the minority side follow 
in a straightforward fashion. 
Well-formedness and individual rationality for the dominant side 
follow directly from \Cref{thm:generic-seller}.

\paragraph{Feasibility.}
By Lemma~\ref{fct:cross-cover-Y}, the aggregate
$Y$ spent by sell-$Y$ users satisfies $Q_Y \ge \crossing$.
Recall that the minority sell-$X$ orders 
supply $D_X$ units of $X$ 
and must receive $\crossing=\rho_0D_X$ units of $Y$. Thus 
the compensation to 
sell-$X$ users is fully funded by the sell-$Y$ orders.
After using $\crossing$ units
of the dominant sell-$Y$ supply to compensate the sell-$X$ users,  the residual
$y:=Q_Y-\crossing\ge0$
units of $Y$ are sent into the AMM's pool, which returns exactly $F(y)$ units of $X$.
The total amount of $X$ available to compensate dominant sell-$Y$ sellers is
therefore
\[
    D_X+F(Q_Y-\crossing)
    =\widetilde F_Y^{\crossing}(Q_Y).
\]
By Equation~\eqref{eq:aggregate-payment-upper}
of \Cref{thm:generic-seller}, 
\begin{equation*}
    P_X := \sum_{i \in \Yside} x_i
    \le
    \widetilde F_Y^{\crossing}(Q_Y)
    =D_X+F(Q_Y-\crossing).
\end{equation*}
where $x_i$ is the amount of $X$ tokens
compensated to user $i \in \Yside$. 
Thus every dominant-side compensation payment is fully funded, and the burnt surplus
    $B_X
    :=D_X+F(Q_Y-\crossing)-P_X
    \ge0$.

\subsection{Pareto Optimality for Eligible Users}
\label{sec:local}

We now prove Pareto optimality for eligible users.
Without loss of generality, we assume
the sell-$Y$ dominated case, since the proof for the other
case is symmetric.

Every eligible sell-$X$ user is on the minority side and, by
\eqref{eq:minority-X}, exhausts its input budget.  Consequently, after the
mechanism executes, no eligible sell-$X$ user has any $X$ left available to
sell.  Any nontrivial trade among eligible users would require both a user
supplying $X$ and a user supplying $Y$.  Thus no subset of eligible users can
make an additional admissible trade among themselves, let alone one that is a
Pareto improvement.

It remains to rule out a strictly beneficial additional trade with the pool.
Let $\reportv{i}\le\sigma_0$ be 
the ask of an eligible sell-$Y$ user. 
Let $y :=Q_Y-\crossing\ge0$
be the residual $Y$ input to the pool after 
compensating the sell-$X$ users where $Q_Y$ is the aggregate $Y$ 
spent by sell-$Y$ users. 
Define the final marginal prices by
\[
    \sigma_e:=F'_+(y)
    \quad\text{($X$ per $Y$)},
    \qquad
    \rho_e:=\frac{1}{\sigma_e}.
\]

An eligible user whose input budget is exhausted cannot make such a trade.
Fix, therefore, a dominant-side sell-$Y$ user $i$  
that is underfilled, i.e., 
$y_i<\reportq{i}$.  
For convenience, let $G(\cdot) := \widetilde F_Y^{\crossing}(\cdot)$. 
If its ask satisfied $\reportv{i}<G'_+(Q_Y)$, 
then increasing $y_i$ by a sufficiently small amount would raise reported
welfare, contradicting the optimality of social welfare
achieved by the allocation rule.  Hence
\[
    \reportv{i}\ge G'_+(Q_Y)
    =F'_+(Q_Y-\crossing)
= \sigma_e
\]
where the first equality follows from the fact that 
$Q_Y\ge \crossing$ and 
the two pieces of \eqref{eq:Ftilde-Y} meet with the same
slope.  
For any additional amount
$0<\delta\le\reportq{i}-y_i$, concavity of $F$ gives
\[
    F(y+\delta)-F(y)
    \le F'_+(y)\delta
    =\sigma_e\delta
    \le\reportv{i}\delta.
\]
The additional pool output is therefore no greater than the user's reported
cost of supplying the additional input, so the trade cannot strictly improve
its utility.  This proves both Pareto-optimality conditions in a sell-$Y$
dominated profile.

\subsection{Incentive Compatibility}
\label{sec:sybil}

It suffices to prove user incentive compatibility (UIC),
since by \Cref{fct:uicimpliesbic}, 
UIC directly gives strategy-proofness
for a builder(-as-user) for our mechanism which is order-insensitive.

We first prove a technical lemma 
which will later help us bound the utility of a strategic user
who submits opposite-direction sybil orders, 
effectively extending the curve's available offering
at the initial price.

\begin{lemma}[Curve-shift lemma]
\label{lem:crossing-shift}
Fix a finite set $O$ of sell-$Y$ reports other than one comparison seller, and
for $\crossing,q\ge0$ let
\begin{equation*}
    \welfareGainY(\crossing;v,q)
    :=
    W_{\widetilde F_Y^{\crossing}}\bigl(O\cup\{(v,q)\}\bigr)
    -W_{\widetilde F_Y^{\crossing}}(O).
\end{equation*}
For every $v\le\sigma_0$, every $\alpha\ge0$, and every comparison cap
$\widehat q$ satisfying $0\le\widehat q\le q+\alpha$,
\begin{equation}
    \welfareGainY(\crossing+\alpha;v,\widehat q)
    \le
    \welfareGainY(\crossing;v,q)+(\sigma_0-v)\alpha.
    \label{eq:shift-Y}
\end{equation}
Symmetrically, if $O$ is a fixed set of sell-$X$ reports and
\[
    \welfareGainX(\crossing;u,r)
    :=
    W_{\widetilde H_X^{\crossing}}\bigl(O\cup\{(u,r)\}\bigr)
    -W_{\widetilde H_X^{\crossing}}(O),
\]
then for $u\le\rho_0$, $\alpha\ge0$, and
$0\le\widehat r\le r+\alpha$,
\begin{equation*}
    \welfareGainX(\crossing+\alpha;u,\widehat r)
    \le
    \welfareGainX(\crossing;u,r)+(\rho_0-u)\alpha.
\end{equation*}
\end{lemma}

\begin{proof}
We prove \eqref{eq:shift-Y} since the other direction is symmetric.  Apply the
layer-cake representation \eqref{eq:seller-layercake}.  Let $D_O(t)$ be the
aggregate capacity in $O$ with ask at most $t$, and let $K_{\crossing}(t)$ be the
threshold quantity associated with the curve $\widetilde F_Y^{\crossing}$.  Extending
the initial spot-price segment from $\crossing$ to $\crossing+\alpha$ translates this threshold
by $\alpha$: for almost every $t\in[0,\sigma_0]$,
\[
    K_{\crossing+\alpha}(t)=K_{\crossing}(t)+\alpha.
\]
For $v\le\sigma_0$, the layer-cake formula gives
\begin{equation*}
    \welfareGainY(\crossing;v,q)
    =
    \int_v^{\sigma_0}
       \min\bigl\{q,(K_{\crossing}(t)-D_O(t))_+\bigr\}\,dt.
\end{equation*}
For every real $w$, every $q,\alpha\ge0$, and every
$0\le\widehat q\le q+\alpha$,
\begin{equation}
    \min\{\widehat q,(w+\alpha)_+\}
    \le
    \min\{q,w_+\}+\alpha.
    \label{eq:min-shift}
\end{equation}
Indeed, replacing $\widehat q$ by the larger $q+\alpha$ can only increase the
left-hand side, after which the inequality is immediate by considering
$w\ge0$, $-\alpha<w<0$, and $w\le-\alpha$.  Apply
Equation~\eqref{eq:min-shift} pointwise with
$w=K_{\crossing}(t)-D_O(t)$ and integrate over
$t\in[v,\sigma_0]$.  The interval has length $\sigma_0-v$, yielding
Equation~\eqref{eq:shift-Y}.
\end{proof}

\subsubsection{A Real Sell-$Y$ User}

Fix all reports of other real users and consider one real sell-$Y$ user of type
$(v,q)$.  Let
\begin{equation*}
   \DYother
    :=\sum_{\substack{i\text{ other sell-}Y\\[2pt] \reportv{i}\le\sigma_0}}\reportq{i},
\qquad
   \DXother
    :=\sum_{\substack{i\text{ other sell-}X\\[2pt] \reportu{i}\le\rho_0}}\reportr{i},
    \qquad
    \MYother:=\rho_0\DXother,
\end{equation*}
where 
$\DYother$
and 
$\DXother$
are the total spot-eligible sell-$Y$ and sell-$X$ capacities of the other
users, respectively.  
Thus 
$\DYother$,
$\DXother$,
and $\MYother$ are fixed before the strategic user's
sybil identities 
are introduced.

\paragraph{Truthful utility.}
The following facts about truthful utility are easy to see: 
\begin{itemize}[leftmargin=6mm,itemsep=1pt]
\item 
If $v>\sigma_0$, truthful utility is zero.  
\item 
If $v\le\sigma_0$ and
$\DYother+q<\MYother$, the truthful user is on the minority side and receives
\begin{equation*}
    \Util^{\rm tr}=(\sigma_0-v)q.
\end{equation*}
Moreover, we have 
\begin{equation}
    \welfareGainY(\MYother;v,q)=(\sigma_0-v)q=\Util^{\rm tr}.
    \label{eq:Delta-equals-minority}
\end{equation}
\item 
If $v\le\sigma_0$ and $\DYother+q\ge \MYother$, truthful sell-$Y$ is dominant and
\begin{equation}
    \Util^{\rm tr}=\welfareGainY(\MYother;v,q)
    \ge(\sigma_0-v)\ell,
    \qquad
    \ell:=(\MYother-\DYother)_+.
    \label{eq:sybil-Y-truth-dom}
\end{equation}
\end{itemize}

Replace the truthful identity by an arbitrary finite set $\playerids$ of sybil bids
in either direction.  Henceforth, we use the aggregate flow notation from
Equation~\eqref{eq:portfolio-flows}.  The user's net $Y$ sale is
$\netYSale=S_Y-R_Y$.  If $\netYSale>q$, the attack has utility $-\infty$ by
Equation~\eqref{eq:portfolio-utility-Y}.  A negative net sale is not 
catastrophic by itself: under Equation~\eqref{eq:portfolio-utility-Y} we must also show
that it entails a strict net loss in $X$.  We verify this separately in each
dominance regime below.  After doing so, 
essentially we only need to consider the regime
\begin{equation}
    0\le \netYSale\le q.
    \label{eq:finite-Y-netcap}
\end{equation}
Notice that Equation~\eqref{eq:finite-Y-netcap} does \emph{not} imply $S_Y\le q$.

\paragraph{\texorpdfstring{The sybil profile is sell-$Y$ dominated.}{The sybil profile is sell-Y dominated.}}
Every spot-eligible sell-$X$ sybil identity is then on the minority side and is
fully executed at the spot price; every ineligible one executes zero. 
Therefore  $R_Y=\rho_0S_X$.
These opposite-direction sybils extend the dominant sell-$Y$ crossing
parameter from $\MYother$ to $\MYother+R_Y$.

The individual payment bound \eqref{eq:individual-payment-upper}, summed over
the user's dominant-side sell-$Y$ sybils, gives $R_X\le\sigma_0 S_Y$.  If
net $Y$ sale $\netYSale = S_Y-R_Y<0$, then $\Delta_X=R_X-S_X=R_X-\sigma_0R_Y
    \le\sigma_0(S_Y-R_Y)<0$. In this case, 
the second catastrophic condition in
\eqref{eq:portfolio-utility-Y} applies, and the attack has $-\infty$ utility.  
Moreover, if 
$\netYSale = S_Y - R_Y> q$, the attack also has $-\infty$ utility.
Thus we may assume 
\begin{equation}
    R_Y\le S_Y \le q+R_Y.
    \label{eq:gross-bound-Ydom}
\end{equation}
The user's utility is
\begin{equation}
\begin{aligned}
    \Util^{\rm syb}
      &=R_X-S_X-v(S_Y-R_Y)\\
      &=R_X-\sigma_0R_Y-v(S_Y-R_Y).
    \label{eq:Yfn-Ydom-utility}
\end{aligned}
\end{equation}

If $v>\sigma_0$, using again $R_X \le\sigma_0 S_Y$ together with
$S_Y-R_Y=\netYSale\ge0$,
\[
    \Util^{\rm syb}
    \le(\sigma_0-v)(S_Y-R_Y)
    \le0
    =\Util^{\rm tr}.
\]

Now suppose $v\le\sigma_0$.  Freeze the curve
$\widetilde F_Y^{\MYother+R_Y}$ and merge the user's sell-$Y$ sybils only
for purposes of comparison.  The UIC guarantee in
part~(e) of Theorem~\ref{thm:generic-seller},
with a virtual sell-$Y$ user of true cost $vz$ and budget equal to the realized gross sale
$S_Y$, gives
\begin{equation}
    R_X-vS_Y
    \le
    \welfareGainY(\MYother+R_Y;v,S_Y).
    \label{eq:Yfn-fixedcurve-netcap}
\end{equation}
By \eqref{eq:gross-bound-Ydom}, $S_Y\le q+R_Y$.  Combining
\eqref{eq:Yfn-Ydom-utility}, \eqref{eq:Yfn-fixedcurve-netcap}, and
Lemma~\ref{lem:crossing-shift},
\begin{align*}
    \Util^{\rm syb}
    &\le
    \welfareGainY(\MYother+R_Y;v,S_Y)
       -(\sigma_0-v)R_Y\notag\\
    &\le
    \welfareGainY(\MYother;v,q).
\end{align*}
The last quantity equals truthful utility both when truth is dominant, by the
VCG marginal-contribution identity, and when truth is minority, by
\eqref{eq:Delta-equals-minority}.  Thus this attack cannot help.

\paragraph{\texorpdfstring{The sybil profile is sell-$X$ dominated.}{The sybil profile is sell-X dominated.}}
The user's active sell-$Y$ sybil identities are now minority identities,  
so they receive $R_X=\sigma_0 S_Y$.
The dominant-side payment bound gives $R_Y \le\rho_0S_X$, or equivalently
$S_X\ge\sigma_0R_Y$.  If $\netYSale = S_Y-R_Y<0$, then
we also have $\Delta_X=\sigma_0 S_Y-S_X
    \le\sigma_0(S_Y-R_Y)<0$.
Hence Equation~\eqref{eq:portfolio-utility-Y} assigns the
attack utility $-\infty$.  
Moreover, if $\netYSale = S_Y-R_Y > q$, the attack also has $-\infty$ utility. 
We may therefore restrict finite-utility attacks in
this regime to $0\le S_Y-R_Y\le q$.  Their utility is
\begin{equation}
    \Util^{\rm syb}=\sigma_0S_Y-S_X-v(S_Y-R_Y)
    \le
    (\sigma_0-v)(S_Y-R_Y)
    =(\sigma_0-v)\cdot (\netYSale).
    \label{eq:Yfn-Xdom-netbound}
\end{equation}
If $v>\sigma_0$, this is nonpositive and cannot beat truthful utility zero.
If $v\le\sigma_0$ and truthful reporting makes sell-$Y$ the minority side,
then $\netYSale\le q$ makes Equation~\eqref{eq:Yfn-Xdom-netbound} at most the truthful value
$(\sigma_0-v)q$.

It remains to consider the case $v\le\sigma_0$ in which truthful reporting
makes sell-$Y$ dominant.  Rewrite the user's utility as
\begin{equation}
    \Util^{\rm syb}
    =(\sigma_0-v)(S_Y-\rho_0S_X)
      +v(R_Y-\rho_0S_X).
    \label{eq:Yfn-Xdom-decomp}
\end{equation}
The second term is nonpositive because $R_Y\le\rho_0S_X$.  Let $\SXother$ be the amount
of $X$ sold, in the sybil outcome, by all other sell-$X$ identities.  
\Cref{fct:cross-cover-X}
gives 
\begin{equation*}
    S_X+\SXother\ge\sigma_0(\DYother+S_Y).
\end{equation*}
Given that $\SXother\le \DXother=\sigma_0\MYother$,
we have 
\[
    S_Y-\rho_0S_X\le \MYother-\DYother\le\ell.
\]
Since $\sigma_0-v\ge0$, Equation~\eqref{eq:Yfn-Xdom-decomp} 
and 
\eqref{eq:sybil-Y-truth-dom} 
yield
\[
    \Util^{\rm syb}
    \le(\sigma_0-v)\ell
    \le \Util^{\rm tr}. 
\]
 Thus no arbitrary finite cross-direction
sybil attack benefits a true sell-$Y$ user under the 
utility function defined in Equation~\eqref{eq:portfolio-utility-Y}.

\subsubsection{A Real Sell-$X$ User}

Fix all reports of other real users and consider one real sell-$X$ user of type
$(u,r)$.  Let
\begin{equation*}
   \DYother
    :=\sum_{\substack{i\text{ other sell-}Y\\[2pt] \reportv{i}\le\sigma_0}}\reportq{i},
\qquad
   \DXother
    :=\sum_{\substack{i\text{ other sell-}X\\[2pt] \reportu{i}\le\rho_0}}\reportr{i},
\qquad
   \MXother:=\sigma_0\DYother.
\end{equation*}
Thus $\MXother$ is the other users' spot-eligible sell-$Y$ capacity measured
in units of $X$.  The proof is obtained from the proof for a real sell-$Y$
user by making the substitutions
\[
\begin{gathered}
    X\longleftrightarrow Y,
    \qquad F\longleftrightarrow H,
    \qquad \sigma_0\longleftrightarrow\rho_0,
    \qquad (v,q)\longleftrightarrow(u,r),\\
    \DYother\longleftrightarrow\DXother,
    \qquad \MYother\longleftrightarrow\MXother,
    \qquad \welfareGainY(M;v,q)\longleftrightarrow\welfareGainX(M;u,r),\\
    S_Y\longleftrightarrow S_X,
    \qquad R_Y\longleftrightarrow R_X,
    \qquad \SXother\longleftrightarrow\SYother,
    \qquad \netYSale\longleftrightarrow\netXSale,
\end{gathered}
\]
and replacing the dominant-side curve $\widetilde F_Y^M$ by
$\widetilde H_X^M$.

The only modifications concern boundary classification, because ties are
resolved in favor of sell-$Y$ dominance.  For a truthful eligible sell-$X$
user, truth is on the minority side when
\[
    \DXother+r\le\MXother,
\]
where equality is included, and its utility is
$(\rho_0-u)r=\welfareGainX(\MXother;u,r)$.  Truth is on the dominant side only when
\[
    \DXother+r>\MXother,
\]
in which case its utility is $\welfareGainX(\MXother;u,r)$ and is at least
\[
    (\rho_0-u)(\MXother-\DXother)_+.
\]
Likewise, a sybil profile is sell-$X$ dominated only under a strict
imbalance $D_Y<\rho_0D_X$.  The reverse branch, $D_Y\ge\rho_0D_X$, includes
equality and is sell-$Y$ dominated.  Thus, at equality, the user's
spot-eligible sell-$X$ sybils are minority orders and execute fully at the
spot rate $\rho_0$, while its ineligible sell-$X$ sybils execute zero.
With these changes to the weak and strict inequalities, every bound in the
sell-$Y$ proof carries over under the substitutions above.

\section{Social Welfare Maximization}
\label{sec:sw}

\elaine{TODO: explain natural, explain why not define welfare using a common numararie,
refer to this in the main scheme section or move thm earlier}

\paragraph{Defining natural.}
We call a deterministic AMM mechanism
\emph{natural} if it satisfies the following properties:

\begin{itemize}[leftmargin=15mm,itemsep=1pt]
\item[\Drop] {\it Drop ineligible bids.} Every ineligible bid receives zero allocation and zero transfer.
Thus a sell-$X$ report with $\reportu{j}>\rho_0$ receives
$(x_j,y_j)=(0,0)$, and a sell-$Y$ report with
$\reportv{i}>\sigma_0$ receives $(y_i,x_i)=(0,0)$.

\item[\NPR] {\it No price reversal.} The residual trade with the AMM does not reverse the direction
implied by the spot-eligible imbalance.  Thus, at a sell-$Y$-dominated profile
the AMM weakly receives $Y$, and at a sell-$X$-dominated profile it weakly
receives $X$.  Additionally, the final marginal price moves weakly against the
dominant side.  In particular, if the profile is sell-$Y$ dominated, then
\[
    \sigma_e\le\sigma_0,
    \qquad\text{and}\qquad
    \rho_e\ge\rho_0.
\]
If the profile is sell-$X$ dominated, then
\[
    \rho_e\le\rho_0,
    \qquad\text{and}\qquad
    \sigma_e\ge\sigma_0.
\]
Here $\sigma_e$ and $\rho_e$ are the final marginal compensation rates in
units of $X$ per $Y$ and $Y$ per $X$, respectively.  At a nondifferentiable
point, the inequalities refer to the marginal rate for an additional trade in
the relevant dominant direction.
\end{itemize}

\paragraph{Justification of naturalness.}
The \Drop condition stipulates that the mechanism simply ignores ineligible bids. This assumption is justified in \Cref{sec:ic-le-impossibility}, where we prove that no UIC mechanism can achieve 
Pareto optimality for all bids including ineligible ones. 
In fact, \Cref{sec:ic-le-impossibility} establishes a stronger 
impossibility result: even an efficiency notion strictly weaker than 
Pareto optimality is incompatible with UIC. \elaine{TODO: double check the section}

\NPR formalizes the conventional price-impact intuition that aggregate
selling pressure on an asset should not cause that asset to appreciate.
This is analogous to the classical Walrasian t\^{a}tonnement principle that prices
adjust in the direction of excess demand~\cite{arrowhurwiczblock}.
\NPR imposes only this directional requirement and makes no assumption
about the magnitude of the resulting price movement.

\begin{lemma}[Structure of minority-side execution]
\label{lem:natural-minority-structure}
Fix a report profile and a deterministic natural mechanism satisfying UIC and
Pareto optimality for eligible users.  Every executed
minority-side order is compensated at the initial spot price.  Moreover,
every minority-side order whose ask is strictly below the initial spot price
exhausts its reported budget.

In particular, at a sell-$Y$-dominated profile, define $M$ by
\begin{equation}
    \sigma_0M=\sum_{j\in\Xside}x_j,
    \label{eq:sw-minority-M-Y}
\end{equation}
so that $\sigma_0M$ is the total $X$ spent by the minority side.  The
minority side receives exactly $M$ units of $Y$, and the dominant side spends
at least $M$ units of $Y$:
\begin{equation}
    \sum_{i\in\Yside}y_i\ge M.
    \label{eq:sw-dominant-covers-minority}
\end{equation}
The symmetric statements hold at a sell-$X$-dominated profile.
\end{lemma}

\begin{proof}
We prove the result for a sell-$Y$-dominated profile.  Consider an eligible
minority-side sell-$X$ user $j$.  If $\reportu{j}<\rho_0$ and
$x_j<\reportr{j}$, then \NPR gives
$\rho_e\ge\rho_0>\reportu{j}$.  The user could sell a sufficiently small
additional amount of $X$ directly to the pool at a marginal rate strictly
above its reservation value, contradicting Pareto optimality for eligible
users.  Hence
\begin{equation}
    x_j=\reportr{j}
    \qquad\text{whenever }\reportu{j}<\rho_0.
    \label{eq:sw-strict-minority-full}
\end{equation}
Fix the other reports and bidder $j$'s reported budget $r$.  For an ask
$t<\rho_0$, let $h(t)$ denote its compensation in $Y$.  By
\eqref{eq:sw-strict-minority-full}, every such report sells $r$.  UIC applied
to any two asks $s,t<\rho_0$, in both directions, implies $h(s)=h(t)=C$.
A type $(t,r)$ can report an ineligible ask and obtain zero by \Drop, so
$C-tr\ge0$; letting $t\uparrow\rho_0$ gives $C\ge\rho_0r$.  Conversely, a
type $(a,r)$ with $a>\rho_0$ obtains zero truthfully and can report
$t<\rho_0$, so UIC gives $0\ge C-ar$; letting $a\downarrow\rho_0$ gives
$C\le\rho_0r$.  Thus $C=\rho_0r$.

For a boundary report $\reportu{j}=\rho_0$, UIC relative to an ineligible
report gives $y_j-\rho_0x_j\ge0$.  A type $a>\rho_0$ that deviates to this
boundary report must obtain utility at most zero, so
$0\ge y_j-ax_j$.  Letting $a\downarrow\rho_0$ gives the reverse inequality.
Consequently,
\begin{equation}
    y_j=\rho_0x_j
    \label{eq:sw-minority-spot-compensation}
\end{equation}
for every eligible minority-side user (and hence for every executed one, by
\Drop).  Summing \eqref{eq:sw-minority-spot-compensation} and using
$\rho_0=1/\sigma_0$ shows that the minority side receives $M$ units of $Y$.
If $y:=\sum_{i\in\Yside}y_i$ is the dominant side's total expenditure, token
balance leaves $y-M$ units of $Y$ as the residual input to the AMM.  By \NPR
this residual trade is in the dominant direction, so $y-M\ge0$.  This proves
\eqref{eq:sw-dominant-covers-minority}.  The sell-$X$-dominated case follows
by symmetry.
\end{proof}

\paragraph{Definition of dominant-side welfare.}
By \Cref{lem:natural-minority-structure}, at a sell-$Y$-dominated outcome the
minority side spends $\sigma_0M$ units of $X$ and receives $M$ units of $Y$,
while the dominant side spends some total $y\ge M$ units of $Y$.  The first
$M$ units of dominant-side expenditure are crossed with the minority supply,
and the remaining $y-M$ units are sent to the pool.  We therefore define
dominant-side social welfare as the $X$ produced by these two sources minus
the reported production cost of the dominant-side $Y$:
\begin{equation}
    \mathsf{DSW}_Y
    :=\sigma_0M+F(y-M)-\sum_{i\in\Yside}\reportv{i}y_i,
    \qquad y:=\sum_{i\in\Yside}y_i.
    \label{eq:dominant-sw-Y}
\end{equation}
For a sell-$X$-dominated outcome, define $M$ by
$\rho_0M=\sum_{i\in\Yside}y_i$ and, writing
$x:=\sum_{j\in\Xside}x_j\ge M$, define symmetrically
\begin{equation*}
    \mathsf{DSW}_X
    :=\rho_0M+H(x-M)-\sum_{j\in\Xside}\reportu{j}x_j.
\end{equation*}
Ineligible bids make zero contribution by \Drop.  The compensation token of
the dominant direction is the numeraire; no arbitrary common numeraire is
imposed across the two dominance regimes.

\begin{theorem}[Welfare maximality among natural mechanisms]
\label{thm:natural-welfare-maximality}
Fix a report profile.  Among deterministic natural mechanisms satisfying UIC
and Pareto optimality for eligible users, our mechanism has
the following properties:
\begin{enumerate}[label=(\roman*),leftmargin=8mm]
\item every strictly eligible minority-side user obtains the same truthful
utility under every mechanism in this class, while every boundary or
ineligible minority-side user obtains zero;
\item if the profile is sell-$Y$ dominated, our mechanism maximizes
$\mathsf{DSW}_Y$; if it is sell-$X$ dominated, our mechanism maximizes
$\mathsf{DSW}_X$.
\end{enumerate}
\label{thm:sw}
\end{theorem}

\begin{proof}
We prove the claims for a sell-$Y$-dominated profile.  The sell-$X$-dominated
case follows by symmetry.

By \Cref{lem:natural-minority-structure}, a strictly eligible sell-$X$ user
$j$ sells its full budget and is compensated at the initial price.  Its
truthful utility is therefore $(\rho_0-u_j)r_j$.  A boundary user has zero
utility whether or not it is executed, and an ineligible user has zero utility
by \Drop.  This proves (i).

\paragraph{Dominant-side welfare.}
Let
\[
    M^*:=\rho_0D_X
\]
be the crossing length used by our mechanism.  It fully executes every
spot-eligible sell-$X$ bid, including boundary bids.  For another natural
mechanism, define $M$ as in \eqref{eq:sw-minority-M-Y}.  The lemma and the
reported caps imply $M\le M^*$; the only possible
difference is partial execution of bids with ask exactly $\rho_0$.

For purposes of the comparison, extend the expression in
\eqref{eq:dominant-sw-Y} to all $y\ge0$ by writing
\[
    \widetilde F^M(y):=
    \begin{cases}
       \sigma_0y,&0\le y\le M,\\
       \sigma_0M+F(y-M),&y\ge M.
    \end{cases}
\]
On every outcome of a mechanism in the comparison class, the lemma gives
$y\ge M$, so the second branch is precisely the welfare expression already
defined in \eqref{eq:dominant-sw-Y}.

We claim that, for every $y\ge0$, the modified compensation curve is
nondecreasing in its crossing length:
\begin{equation}
    M\le M^*
    \quad\Longrightarrow\quad
    \widetilde F^M(y)\le\widetilde F^{M^*}(y).
    \label{eq:sw-crossing-monotonicity}
\end{equation}
If $y\le M$, the two sides are equal to $\sigma_0y$.  If
$M<y\le M^*$, concavity and $F'_+(0)=\sigma_0$ give
$F(y-M)\le\sigma_0(y-M)$, so the claim follows.  If $y>M^*$,
concavity gives
\[
    F(y-M)-F(y-M^*)
    \le \sigma_0(M^*-M),
\]
which again proves \eqref{eq:sw-crossing-monotonicity}.

Let $(y_i)$ be the dominant-side allocation of the other mechanism.  It is a
feasible candidate for our mechanism's one-sided welfare problem because the
reported caps are the same.  Our allocation maximizes the one-sided objective
with curve $\widetilde F^{M^*}$, and
\eqref{eq:sw-crossing-monotonicity} therefore gives
\begin{align*}
    \mathsf{DSW}_Y^{\rm ours}
    &\ge
      \widetilde F^{M^*}\!\left(\sum_i y_i\right)
        -\sum_i\reportv{i}y_i\\
    &\ge
      \widetilde F^{M}\!\left(\sum_i y_i\right)
        -\sum_i\reportv{i}y_i
     =\mathsf{DSW}_Y^{\rm other}.
\end{align*}
This proves part (ii).  Exchanging $X$ with $Y$, $F$ with $H$, and
$\sigma_0$ with $\rho_0$ proves both claims for a sell-$X$-dominated profile.
\end{proof}

\elaine{TODO: can we remove the naturalness assumption or the deterministic assumption?}

\begin{remark}[Weak local efficiency suffices]
In fact, \Cref{lem:natural-minority-structure,thm:natural-welfare-maximality}
continue to hold even if Pareto optimality for eligible users is replaced by the
weaker notion of \emph{weak local efficiency} defined in
\cite{ammtrilemma}, which only requires that eligible users cannot obtain a
Pareto improvement by trading with the pool.  Indeed, the only use of Pareto
optimality in the proofs is to show that a strictly eligible minority-side
user with unspent budget could profitably make an additional trade with the
pool.  Weak local efficiency rules out exactly this possibility.
\end{remark}

\section{Restrictions on Builder Fee Structure}
\label{sec:builder-impossibility}

Our mechanism burns the unused compensation token.  A natural
alternative is to transfer some of that excess to the block builder and burn
only the remainder.  
In this section, we show that subject to UIC and strategy-proofness
for a builder who is simultaneously a user with intrinsic demand,  
the only possible builder fee structure
is to have the residual surplus entirely ``burnt'', 
i.e., redistributed to the community, 
and pay the builder zero fees. 
In practice, however, one can remunerate 
the builder with a constant exogenous block subsidy
that is independent of the current batch's outcomes, 
e.g., using the community surplus accmulated so far. 
Such a payment will not affect
the game theoretic guarantees achieved by the mechanism. 

Our ``zero builder fee'' impossibility 
can be viewed as a seller-side analogue of 
the ``zero miner revenue'' impossibility in transaction-fee mechanisms; see,
for example, Theorem~4.7 of Chung and Shi~\cite{foundation-tfm}.  
Our proof techniques are also inspired by their work. 

\subsection{Additional Preliminaries and Notations}

Because this is an impossibility result, it is enough to restrict attention to
a smaller type and strategy space.  
Specifically, we consider a 
a finite sell-$Y$-only batch. 
Moreover, since our proof will only need strategies that misreport valuation, we may
fix arbitrary budget $\bar q_i>0$.  
Thus we effectively consider a single-parameter environment
where only 
the per-unit reservation values are private.  
Any mechanism that is incentive compatible on the
richer type space of the main model must remain incentive compatible on this
single-parameter restriction.

Although our earlier feasibility is a deterministic
AMM mechanism, our impossibility result in this section
holds even for {\it randomized} mechanisms.  For a report vector
$\widetilde{\mathbf v}=(\reportv{1},\ldots,\reportv{n})$, let
\[
    \bar y_i(\widetilde{\mathbf v})\in[0,\bar q_i]
\]
denote the {\it expected} amount of $Y$ sold by user $i$, and let
$\bar x_i(\widetilde{\mathbf v})$ denote its {\it expected} compensation in
$X$.  Let
\[
    \bar\mu(\widetilde{\mathbf v})\ge0,
\]
and $\bar\beta(\widetilde{\mathbf v})\ge0$ denote the {\it expected} builder fee and
{\it expected} burn amount, respectively.  
In the above, all expectations are taken over the mechanism's random coins.
Let
\[
    \bar Y(\widetilde{\mathbf v}):=\sum_i \bar y_i(\widetilde{\mathbf v}).
\]

In our impossibility result (\Cref{thm:zero-builder-fee}), we will make
use of individual rationality (IR), feasibility, user incentive compatibility (UIC), and 
strategy-proofness for a builder who is also a user with intrinsic demand. 
Since we want the impossibility to work even for 
possibly randomized mechanisms, the UIC and strategy-proofness
definitions will be modified to use the expected utility instead. 
Note that for a builder-as-user, the fees obtained by   
the builder will be added to its utility. 
For a feasibility result we would ideally
want IR to hold with probability $1$. However, for our impossibility
proof, we only need IR to hold in the expectation too, that is,   
$\bar x_i-v_i\bar y_i\ge0$.
Feasibility requires that 
the builder fees be entirely funded only from the current batch's available AMM
output.  Ex-post feasibility and concavity of $F$ imply the following expected
accounting constraint:
\begin{fact}[Accounting constraint imposed by feasibility] 
For a possibly randomized AMM mechanism with a concave pricing curve that satisfies
feasibility,  
the following accounting constraint must hold when there are only sell-$Y$ bids. 
\begin{equation}
    \sum_i \bar x_i(\widetilde{\mathbf v})+\bar\mu(\widetilde{\mathbf v})
       +\bar\beta(\widetilde{\mathbf v})
    \le F\!\left(\bar Y(\widetilde{\mathbf v})\right).
    \label{eq:builder-funding}
\end{equation}
\label{fct:builder-funding}
\end{fact}
\begin{proof}
Feasibility requires that 
with probability 1, the compensation to all users, the builder's fees, 
and the burnt amount are all fully funded by  
AMM pool's output. 
Taking expectation, and applying Jensen's inequality using the concavity
of $F$, 
we immediately get the stated accounting constrant. 
\end{proof}

Throughout, we assume that the mechanism 
is not aware which identities are owned by the builder.

\subsection{Proof of Zero Builder Fee}

\begin{lemma}[Seller payment uniqueness]
\label{lem:seller-payment-uniqueness}
Fix a sell-$Y$ user $i$ and the reports of all other users.  Let
$\bar y_i(\reportv{i})$ be user $i$'s expected amount of $Y$ spent when it reports
reservation value $\reportv{i}$.  Suppose two expected compensation rules
$\bar x_i(\reportv{i})$ and $\bar x_i'(\reportv{i})$ both make truthful reporting
dominant for every true reservation value $v_i$, under utilities
    $\bar x_i(\reportv{i})-v_i\bar y_i(\reportv{i})$ and
    $\bar x_i'(\reportv{i})-v_i\bar y_i(\reportv{i})$,
respectively.  Then
$\bar x_i'(\reportv{i})-\bar x_i(\reportv{i})$ is constant in $\reportv{i}$.
\end{lemma}

\begin{proof}
This is the seller form of the usual single-parameter payment-uniqueness
lemma.  Dominant-strategy truthfulness implies that $\bar y_i(\reportv{i})$ is
nonincreasing in $\reportv{i}$.  For the truthful compensation rule $\bar x_i$,
the envelope identity~\cite{myerson,archer2001truthful} (i.e.,
the unique payment identity of Myerson's lemma adapted
to a seller or procurement setting) 
gives
\[
    \bar x_i(v_i)-v_i\bar y_i(v_i)
    =\bar x_i(0)-\int_0^{v_i}\bar y_i(s)\,ds,
\]
and hence
\begin{equation*}
    \bar x_i(v_i)
    =v_i\bar y_i(v_i)+\bar x_i(0)-\int_0^{v_i}\bar y_i(s)\,ds.
\end{equation*}
The same identity holds for $\bar x_i'$ with exactly the same spending
rule $\bar y_i$.  Subtracting the two identities yields
\[
    \bar x_i'(v_i)-\bar x_i(v_i)
    =\bar x_i'(0)-\bar x_i(0),
\]
which is independent of $v_i$.  The argument applies to expected spending and
expected compensation for a randomized mechanism.
\end{proof}

\begin{theorem}[Zero builder fees]
\label{thm:zero-builder-fee}
Let $F:[0,+\infty)\to\R_{\ge0}$ be continuous, strictly increasing, and
concave, with $F(0)=0$, and let $\sigma_0=F'_+(0)\in(0,+\infty)$. 
Consider any possibly randomized mechanism 
that satisfies individual rationality, feasibility, 
UIC, and strategy-proofness for a builder-as-user. 
Then, for every report profile $\widetilde{\mathbf v}$,
\begin{equation}
    \bar\mu(\widetilde{\mathbf v})=0.
    \label{eq:zero-builder-fee}
\end{equation}
In particular, because the builder fee is nonnegative, zero
expectation implies that it is zero with probability $1$.  
\end{theorem}

As mentioned, this impossibility
uses only sell-$Y$ profiles and only deviations
that misreport one's valuation. Thus the impossibility
continues to hold when allowing
profiles in both directions and a broader strategy space, 
e.g., when 
users and builders can misreport both their valuation and budget, and 
when a builder may additionally censor or reorder bids
or inject sybil orders.

\begin{proof}(of \Cref{thm:zero-builder-fee}.)
Fix a user identity $i$ and all reports $\widetilde{\mathbf v}_{-i}$.  For a possible report $\reportv{i}$ of
identity $i$, abbreviate $\bar y_i(\widetilde{\mathbf v}_{-i},\reportv{i})$,
$\bar x_i(\widetilde{\mathbf v}_{-i},\reportv{i})$, and
$\bar\mu(\widetilde{\mathbf v}_{-i},\reportv{i})$ by
$\bar y_i(\reportv{i})$, $\bar x_i(\reportv{i})$, and $\bar\mu(\reportv{i})$,
respectively.  By user incentive compatibility, $\bar x_i(\reportv{i})$ is a
truthful seller compensation rule for spending rule $\bar y_i(\reportv{i})$.  Now
consider the same numerical bid profile in the
admissible role assignment in which identity $i$ belongs to the block builder.
Strategy-proofness for a builder-as-user says that
\[
    \bar x_i(\reportv{i})+\bar\mu(\reportv{i})
\]
is also a truthful seller compensation rule for the \emph{same} spending rule
$\bar y_i(\reportv{i})$, because the builder's total utility from this restricted deviation is
\[
    \bar x_i(\reportv{i})+\bar\mu(\reportv{i})-v_i\bar y_i(\reportv{i}).
\]
Lemma~\ref{lem:seller-payment-uniqueness} therefore implies that
\begin{equation}
    \bar\mu(\reportv{i})\ \text{is independent of }\reportv{i}
    \quad\text{when }\widetilde{\mathbf v}_{-i}\text{ is fixed}.
    \label{eq:builder-coordinate-independence}
\end{equation}
Since strategy-proofness for a builder-as-user applies no matter which
identity is builder-owned,
\eqref{eq:builder-coordinate-independence} holds for every coordinate.  Thus
expected builder revenue is invariant under changing the bids one coordinate
at a time.  For every finite profile $\widetilde{\mathbf v}$,
\begin{equation}
    \bar\mu(\widetilde{\mathbf v})
    =\bar\mu(\sigma_0,\ldots,\sigma_0).
    \label{eq:builder-anchor-reduction}
\end{equation}

It remains to evaluate the anchor profile in which every user's true
reservation value and truthful report equal $\sigma_0$.  User individual
rationality gives, for every $i$,
\[
    \bar x_i\ge \sigma_0\bar y_i.
\]
Summing and writing $\bar Y=\sum_i\bar y_i$,
\begin{equation}
    \sum_i\bar x_i
    \ge\sigma_0\bar Y.
    \label{eq:builder-anchor-ir}
\end{equation}
Concavity, $F(0)=0$, and $F'_+(0)=\sigma_0$ imply the global tangent bound
\begin{equation}
    F(y)\le\sigma_0 y 
    \qquad\text{for every } y\ge0.
    \label{eq:builder-tangent-bound}
\end{equation}
Using the funding constraint \eqref{eq:builder-funding} and
\eqref{eq:builder-anchor-ir}--\eqref{eq:builder-tangent-bound},
\begin{align*}
    \sigma_0\bar Y
    +\bar\mu
    +\bar\beta
    &\le
    \sum_i\bar x_i
      +\bar\mu
      +\bar\beta\\
    &\le
    F(\bar Y)\\
    &\le
    \sigma_0\bar Y.
\end{align*}
Hence
\[
    \bar\mu+\bar\beta\le0.
\]
Both quantities are nonnegative, so in particular
\[
    \bar\mu(\sigma_0,\ldots,\sigma_0)=0.
\]
Combining this with \eqref{eq:builder-anchor-reduction} proves
\eqref{eq:zero-builder-fee}.
\end{proof}

\section{Impossibility with Censorship}
\label{sec:censorship-impossibility}

For our feasibility result, we assume that the underlying consensus 
guarantees censorship resilience, and thus the ability to censor
is excluded from the builder's strategy space.  
In this section, we justify this assumption by
proving an impossibility. 
We show that when the block size is finite and the block builder may censor bids,  
no AMM mechanism
can simultaneously satisfy individual rationality, feasibility,
UIC, as well as strategy-proofness for a builder
who is simulatenously a user with intrinsic demand.
Further, this impossibility 
holds even for randomized mechanisms. 
The proofs in this section are inspired
by the transaction fee mechanism literature, e.g., see
Corollary 4.9 of Chung and Shi~\cite{foundation-tfm}.

\paragraph{A finite block model allowing censorship.}
More specifically, let $k$ denote the finite block size. The builder is tasked with selecting up to $k$ bids from among all outstanding bids and including them in the block in some prescribed order. 
Then, some trusted on-chain rule makes allocation
and compensation decisions based on the included bids and their relative ordering.

A strategic builder-as-user may observe the other users' bids and then strategically choose which bids to include and how to sequence them. In particular, it may omit any pending bids and may inject and include bids associated with its own sybil identities. 

The strategy space of a strategic user remains the same as before.
As in \Cref{sec:builder-impossibility}, we assume
that the mechanism is unaware which user identities
are owned by the builder.

\paragraph{Anonymity.}
We assume that the AMM mechanism is anonymous: its on-chain rules make allocation and payment decisions based solely on the valuations and budgets reported in the bids, as well as their relative ordering within the block. Any identity-related information (e.g., public keys) is ignored. This anonymity assumption is both desirable and satisfied by real-world AMMs, as well as by the mechanisms proposed in this paper. In particular, anonymity is consistent with the core decentralization principles underlying modern blockchains.

\paragraph{Non-triviality.}
We call the mechanism \emph{non-trivial} if there is some block
$S$ and an identity $i\in S$ whose expected compensation is strictly positive:
$\mathbb E[x_i(S)]>0$.  The expectation is over all randomization of the
mechanism.

\begin{theorem}[Impossibility with censorship]
\label{thm:censorship-impossibility}
Assume that the block size is finite. 
No non-trivial, anonymous AMM mechanism
can simultaneously satisfy
individual rationality, feasibility, UIC, and
strategy-proofness for a builder-as-seller.  
\end{theorem}

\begin{proof}
Suppose, toward a contradiction, that the mechanism is non-trivial.  Then
some block $S$, of size at most $k$, contains a bid $j$ with positive expected
compensation.  Write
\[
    T:=\mathbb E[x_j(S)]>0,
    \qquad
    A:=\mathbb E[y_j(S)],
\]
and let $q\ge0$ be the public quantity cap of that bid.  If $q>0$, choose
\begin{equation*}
    0<\varepsilon<\frac{T}{2q}.
\end{equation*}
If $q=0$, choose any $\varepsilon>0$; the cost of every feasible allocation
to an added identity is then zero.

\paragraph{Create many low-ask identities.}
Construct a pending-bid set containing all identities in $S\setminus
\{j\}$ and $N$ additional identities, each with cap $q$, true reservation
value $\varepsilon$, and truthful ask $\varepsilon$.  Run the prescribed
block-construction strategy on this pool.  For each added identity $i$, let $\Util_i$ denote
its expected trading utility under truthful behavior; an omitted bid has
allocation and compensation zero.  Individual rationality gives
$\Util_i\ge0$.

\paragraph{One identity has small truthful utility.}
Finite block capacity bounds the added identities' total compensation independently
of their number.  Indeed, if
$Q_S$ is the total budget of the bids in $S\setminus\{j\}$, a realized
block contains at most $k$ bids and hence has total budget at most
$Q_S+kq$.  Feasibility, nonnegative builder fee and burn, and monotonicity of
$F$ therefore give
\[
    \sum_{i=1}^N x_i
    \le \sum_{h\text{ included}}x_h
    \le F(Q_S+kq)=:C.
\]
Since each added identity has nonnegative cost,
\[
    \sum_{i=1}^N \Util_i
    =\mathbb E\!\left[\sum_{i=1}^N
       (x_i-\varepsilon y_i)\right]
    \le C.
\]
Consequently some identity $i^*$ satisfies $\Util_{i^*}\le C/N$.

\paragraph{Censorship gives that identity a profitable deviation.}
Now imagine that the unlucky user $i^*$ belongs to 
the builder. 
Its builder fee is zero
with probability $1$ by
Theorem~\ref{thm:zero-builder-fee}, 
so the builder's truthful utility is $\Util_{i^*}$.  
The builder can deviate by censoring the prescribed block, including
$S\setminus\{j\}$ together with $i^*$, and making $i^*$ submit $j$'s report.
The on-chain rule then sees the same numerical block $S$, with $i^*$ in place
of $j$.  Hence the deviation gives the builder expected utility
\[
    T-\varepsilon A
    \ge T-\varepsilon q
    >\frac{T}{2},
\]
again with zero builder fee.  Taking $N>2C/T$ makes the truthful utility
$\Util_{i^*}\le C/N<T/2$, contradicting builder-as-seller incentive
compatibility.  Therefore no block can give any seller positive expected
compensation.
\end{proof}

\elaine{can we weaken the anonymity assumption?}

\section*{Acknowledgements}
Human intelligence was solely responsible for the 
new conceptual contributions including the proposal of this new paradigm.  
The main mechanism and some of its proofs were developed with assistance from ChatGPT 5.6. 
Other results were first proved with human intelligence and subsequently written up with assistance 
from ChatGPT 5.6 and Claude Fable 5. The authors carefully reviewed and substantially 
edited all AI-generated text.

\bibliographystyle{alpha}
\bibliography{refs}

\appendix

\section{Why Exclude Ineligible Users?}
\label{sec:ic-le-impossibility}

Recall that our Pareto optimality (PO) notion is required
to hold only for eligible users whose asks
are no higher than the initial spot price. 
We now explain why we do not aim to achieve efficiency for ineligible users. 
Specifically, we prove that 
the requirement of 
UIC precludes achieving PO for all users, including ineligible ones. 
Li, Shi, and Zhang~\cite{ammtrilemma} proved a similar impossibility
but their setting requires that the mechanism's outcome
land on the pricing curve without burning (or redistribution) of the surplus. 
Therefore, their proof cannot be easily adapted to our setting
where part of the surplus may be burnt.  
Our new argument below instead
uses the information rents forced by UIC, 
i.e., 
the extra utility a user obtains because the mechanism must remain incentive compatible
  across different private reservation values. 

\paragraph{Two-user scenario for the proof.}
It suffices to consider one sell-$Y$ user $i$ and one sell-$X$ user $j$.  In
the notation of the model, user $i$ supplies $y_i$ units of $Y$ and receives
$x_i$ units of $X$, while user $j$ supplies $x_j$ units of $X$ and receives
$y_j$ units of $Y$.
Feasibility without an outside subsidy means that there is a signed residual
AMM input $p$ and nonnegative surplus-disposal amounts $B_Y,B_X$ such that
\begin{equation}
    y_i=y_j+p+B_Y,
    \qquad
    x_j+F(p)=x_i+B_X.
    \label{eq:ic-le-offcurve-accounting}
\end{equation}
Here $p$ units of $Y$ are sent to the AMM and $F(p)$ units of $X$ are released.
The amounts $B_Y$ and $B_X$ may be burned or redistributed outside the current
batch.  In particular, \eqref{eq:ic-le-offcurve-accounting} only implies
\begin{equation}
    y_i\ge y_j+p,
    \qquad
    x_i\le x_j+F(p),
    \label{eq:ic-le-accounting-bounds}
\end{equation}
and does not require the aggregate user outcome itself to lie on the curve.
As elsewhere in the paper, the same residual trade determines the final
marginal prices
\begin{equation*}
    \sigma_e:=F'_+(p),
    \qquad
    \rho_e:=\frac{1}{\sigma_e}.
\end{equation*}

\begin{theorem}[UIC precludes Pareto optimality for all users]
\label{thm:ic-le-impossibility}
Let $F$ be the finite-reserve concave AMM pricing curve defined in
\Cref{sec:curve}.
No deterministic AMM mechanism satisfying user individual rationality and the
off-curve feasibility condition
\eqref{eq:ic-le-offcurve-accounting} can simultaneously satisfy user incentive
compatibility (UIC) and Pareto optimality for all users, including ineligible
users.
In fact, the impossibility already holds for two users with fixed, publicly known
budgets and with deviations restricted to reservation-value misreports.  
\end{theorem}

\begin{proof}
Consider the two-user scenario mentioned above.  
Specifically, 
for each $v\in[v_0,v_1)$, consider a truthful sell-$Y$ user of type $(v,Q)$
and a truthful sell-$X$ user of type $(u,R)$.  Denote their respective
allocations and compensations by
\[
    (x_i(v),y_i(v))
    \qquad\text{and}\qquad
    (x_j(v),y_j(v)),
\]
and denote the residual AMM input by $p(v)$.
We next explain how we choose the parameters $v_0$, $v_1$, $Q$ and $R$.
Let 
\begin{equation*}
    0<v_0<v_1<\sigma_0,
    \qquad
    \delta:=1-\frac{v_0}{v_1}>0.
\end{equation*}

Define the maximum AMM surplus at per-unit cost $v_0$ by
\begin{equation}
    \Gamma(v_0)
       :=\sup_{p\ge0}\bigl\{F(p)-v_0p\bigr\}.
    \label{eq:ic-le-gamma}
\end{equation}
Because the pool has finite $X$ reserves, $F(p)<x_0$ for every $p\ge0$.
Consequently $\Gamma(v_0)\le x_0<\infty$.  Choose the sell-$X$ budget $R$
large enough that the eventual rent lower bound $R\delta$ exceeds this
maximum surplus, and then choose the sell-$Y$ budget $Q$ large enough that a
full sell-$Y$ allocation cannot be funded (to be shown later):
\begin{equation*}
    R>\frac{\Gamma(v_0)}{\delta},
    \qquad
    v_0Q>R+F(Q).
\end{equation*}
The second choice is possible by taking $Q>(R+x_0)/v_0$.  Finally, set
\begin{equation*}
    u:=\frac{1}{v_1}.
\end{equation*}
Since $v_1<\sigma_0$, this choice gives
$u=1/v_1>1/\sigma_0=\rho_0$.  Thus the sell-$X$ user with valuation $u$ is ineligible at
the initial spot price.  Requiring Pareto optimality for this user is precisely
what will drive the contradiction.

\paragraph{Finite reserves force the sell-$Y$ user to be underfilled.}
Suppose, towards a contradiction, that $y_i(v)=Q$.  Individual rationality gives
\begin{equation}
    x_i(v)\ge vQ\ge v_0Q.
    \label{eq:ic-le-y-ir}
\end{equation}
Individual rationality of the sell-$X$ user also gives
$y_j(v)\ge u x_j(v)\ge0$.
The first identity in \eqref{eq:ic-le-offcurve-accounting} therefore implies
$p(v)\le y_i(v)=Q$.  Since $x_j(v)\le R$ and $F$ is increasing, the second
accounting bound gives
\[
    x_i(v)\le x_j(v)+F(p(v))\le R+F(Q)<v_0Q,
\]
contradicting \eqref{eq:ic-le-y-ir}.  Consequently,
\begin{equation*}
    y_i(v)<Q.
\end{equation*}

\paragraph{Pareto optimality fully executes the sell-$X$ order.}
Because the sell-$Y$ user is underfilled, it could profitably sell an
additional amount directly to the pool if $\sigma_e>v$.  Pareto optimality
for all users rules this out, and hence
\begin{equation*}
    F'_+(p(v))=\sigma_e\le v<\sigma_0.
\end{equation*}
Concavity and differentiability at zero imply
$F'_+(p)\ge F'(0)=\sigma_0$ whenever $p\le0$.  Hence
\begin{equation}
    p(v)>0.
    \label{eq:ic-le-positive-p}
\end{equation}
Moreover,
\[
    \rho_e=\frac{1}{F'_+(p(v))}
       \ge\frac{1}{v}
       >\frac{1}{v_1}=u.
\]
Thus underfilling the sell-$Y$ user pushes the marginal compensation in the
reverse direction above the sell-$X$ user's ask.
If the sell-$X$ user were underfilled, the strict inequality $\rho_e>u$
would likewise give it a profitable additional trade with the pool, contrary
to Pareto optimality for all users.  Therefore
\begin{equation}
    x_j(v)=R.
    \label{eq:ic-le-x-full}
\end{equation}

\paragraph{UIC forces a large sell-$X$ compensation.}
Fix $v$ and vary only the sell-$X$ user's true and reported ask $a$ over
$0\le a<1/v$, while keeping its budget equal to $R$.  The preceding underfilling
argument does not depend on the numerical sell-$X$ ask.  For every such $a$,
the same pool-trade condition therefore forces its allocation to equal $R$.

Let $y_j(a)$ be its compensation.  For any $a,b<1/v$, the two pairwise UIC
inequalities are
\[
    y_j(a)-aR\ge y_j(b)-aR,
    \qquad
    y_j(b)-bR\ge y_j(a)-bR.
\]
Because every ask in this interval receives the same full allocation, these
inequalities force the compensation to be constant: $y_j(a)=y_j(b)$.
Individual
rationality gives $y_j(a)\ge aR$ for every $a<1/v$; taking $a\uparrow1/v$ shows
that this constant is at least $R/v$.  Since the actual ask
$u=1/v_1$ lies in the interval, we obtain
\begin{equation}
    y_j(v)\ge\frac{R}{v}.
    \label{eq:ic-le-s-lower}
\end{equation}
Intuitively, the sell-$X$ user must receive at least $R/v$ units of $Y$.
Those units, together with the positive residual input $p(v)$ sent to the
pool, must be supplied by the sell-$Y$ user.  Accordingly, combining
\eqref{eq:ic-le-accounting-bounds},
\eqref{eq:ic-le-positive-p}, and \eqref{eq:ic-le-s-lower} gives
\begin{equation}
    y_i(v)\ge y_j(v)+p(v)
       \ge\frac{R}{v}+p(v).
    \label{eq:ic-le-y-lower}
\end{equation}

\paragraph{UIC forces a large sell-$Y$ information rent.}
Let
\[
    U(v):=x_i(v)-v y_i(v)
\]
be the sell-$Y$ user's truthful utility.  UIC for a user
of true type $(v_0,Q)$, comparing truthful reporting with the report $(v,Q)$,
implies
\begin{align*}
    U(v_0)
       &\ge x_i(v)-v_0y_i(v) \notag\\
       &=U(v)+(v-v_0)y_i(v).
\end{align*}
The low-cost type $v_0$ can mimic type $v$.  UIC therefore requires it to
retain the cost advantage $(v-v_0)y_i(v)$ as information rent.
Individual rationality gives $U(v)\ge0$, while
\eqref{eq:ic-le-y-lower} gives $y_i(v)\ge R/v$.  Hence, for every
$v\in(v_0,v_1)$,
\[
    U(v_0)\ge R\left(1-\frac{v_0}{v}\right).
\]
Taking the supremum as $v\uparrow v_1$ yields
\begin{equation}
    U(v_0)\ge R\delta.
    \label{eq:ic-le-rent-lower}
\end{equation}

\paragraph{The forced rent exceeds the AMM's available surplus.}
Write $p_0:=p(v_0)$ and similarly use a subscript $0$ for the other outcomes
at $v_0$.  From \eqref{eq:ic-le-y-lower} and
\eqref{eq:ic-le-rent-lower},
\begin{align*}
    x_{i,0}
       &=v_0y_{i,0}+U(v_0) \notag\\
       &\ge v_0\left(\frac{R}{v_0}+p_0\right)+R\delta \notag\\
       &=R+v_0p_0+R\delta.
\end{align*}
On the other hand, \eqref{eq:ic-le-x-full} and off-curve feasibility imply
\begin{equation*}
    x_{i,0}\le R+F(p_0).
\end{equation*}
The lower bound includes the sell-$Y$ user's production cost and its forced
information rent, whereas the upper bound is all the $X$ available from the
sell-$X$ user and the AMM.  Combining them gives
\[
    F(p_0)-v_0p_0
       \ge R\delta
       >\Gamma(v_0),
\]
whereas $p_0>0$ by \eqref{eq:ic-le-positive-p}.  This contradicts the
definition of $\Gamma(v_0)$ in \eqref{eq:ic-le-gamma}.

The argument used only single-identity reservation-value UIC with fixed
budgets, which is a restriction of the full UIC requirement.
The claimed impossibility follows.
\end{proof}

\begin{remark}[The proof only needs local efficiency]
The proof does not use the part of Pareto optimality that rules out Pareto
improvements through trades among users.  It uses only the requirement that
no user can trade directly with the pool to obtain a Pareto improvement.
This weaker condition is called \emph{local efficiency} (LE)~\cite{ammtrilemma}.  If
$\sigma_e$ and $\rho_e=1/\sigma_e$ are the final marginal compensation rates,
LE requires
the following to hold for all users $i$: 
\begin{equation*}
    y_i<\reportq{i}
       \ \Longrightarrow\ \reportv{i}\ge \sigma_e,
    \qquad
    x_j<\reportr{j}
       \ \Longrightarrow\ \reportu{j}\ge \rho_e.
\end{equation*}
Thus the theorem remains true if Pareto optimality for all users is replaced
by LE: UIC is incompatible even with this weaker efficiency requirement.
\end{remark}

\section{Necessity of Surplus Redistribution}
\label{sec:zero-burn-impossibility}

We now show that surplus burning is essential: without it, no AMM mechanism
can simultaneously provide the full set of guarantees achieved by our
mechanism---UIC, strategy-proofness for
a builder with or without intrinsic demand, and Pareto optimality for eligible
users.  We stress, however, surplus redistribution is more than a technical device for
circumventing this impossibility.  As discussed in \Cref{sec:intro}, it offers
a compelling new paradigm for provably eliminating MEV in two-asset AMM
contracts while substantially reducing 
the negative externalities of today's highly
centralized builder ecosystem~\cite{buildercentral}.

Li, Shi, and Zhang~\cite{ammtrilemma} showed that, when surplus burning is
prohibited, UIC and Pareto optimality (PO) for {\it all} users cannot be
achieved simultaneously.  However, as shown in
\Cref{sec:ic-le-impossibility}, requiring PO for all users is too stringent:
the impossibility persists even when surplus burning is allowed.
Our goal here is instead to establish the necessity of surplus burning
even under the weaker requirement of PO for {\it eligible} users only.
This distinction calls for a different proof argument described below.

For a concave function, write $F'_-(z)$ and $F'_+(z)$ for its left and right
derivatives.  Both exist at every interior point, and
\begin{equation}
    s<t \quad\Longrightarrow\quad F'_-(t)\le F'_+(s).
    \label{eq:zero-burn-slope-order}
\end{equation}

\begin{theorem}[Zero burn is impossible]
\label{thm:zero-burn-impossibility}
Let $F$ satisfy assumptions \textnormal{(A1)--(A4)} of
\Cref{sec:curve}.  
Suppose that a truthful zero-budget user must receive zero utility, then 
no deterministic AMM mechanism can simultaneously satisfy
individual rationality, feasibility with zero burn, UIC, strategy-proofness
for a builder-as-user, and Pareto optimality for eligible bids. 

Here zero-burn feasibility means exact token accounting: every token made
available by opposite-side orders and the residual AMM trade is paid either
to users or to the builder, and the burn amount must be zero.
\end{theorem}

\begin{proof}
Suppose, toward a contradiction, that such a mechanism exists.  We use only
sell-$Y$ profiles.  The key tension is between marginal and average
compensation: Pareto optimality determines who trades by comparing asks with
the pool's marginal compensation, whereas zero burn can force a sole trading
user to receive the pool's entire output.  We will exploit the gap between
these two rates to make an otherwise excluded user profit by understating
its valuation.

By \Cref{thm:zero-builder-fee}, strategy-proofness for a
builder-as-user, together with IR, feasibility, and UIC, forces the builder
fee to be zero on every profile.  Since the burn is also zero, a sell-$Y$-only
outcome with total input $Y:=\sum_i y_i$ satisfies the exact accounting
identity
\begin{equation}
    \sum_i x_i=F(Y).
    \label{eq:zero-burn-exact-accounting}
\end{equation}

Due to \Cref{fct:uic-implies-nfl}, we have the following: 
\begin{equation}
    y_i=0 \quad\Longrightarrow\quad x_i=0
    \qquad\text{on a sell-$Y$-only profile.}
    \label{eq:zero-burn-zero-allocation}
\end{equation}

We first record the allocation restrictions imposed by Pareto optimality:
cheaper supply must be used first, and users must have no profitable way to
increase or unwind their pool trades.  Consider a truthful sell-$Y$-only
profile consisting of eligible users and put $Y=\sum_i y_i$.

First, if users $i$ and $j$ have asks $v_i<v_j$, and $j$ sells a positive
amount, then $i$ must exhaust its cap.  Otherwise, for a sufficiently small
$\varepsilon>0$, user $i$ can give $\varepsilon$ units of $Y$ to user $j$ in
exchange for $r\varepsilon$ units of $X$, where
$v_i<r<v_j$.  Choose $\varepsilon$ small enough that user $j$ returns no more
than the $X$ and $Y$ involved in its original trade.  User $i$'s utility rises
by $(r-v_i)\varepsilon$, and user $j$'s utility rises by
$(v_j-r)\varepsilon$.  The pool position and every other user's outcome are
unchanged.  This is a Pareto improvement among eligible users.  Therefore
\begin{equation}
    v_i<v_j,\ y_j>0
       \quad\Longrightarrow\quad y_i=q_i.
    \label{eq:zero-burn-priority}
\end{equation}

Second, any user who sells a positive amount must have ask at most the final
left marginal compensation:
\begin{equation}
    y_i>0 \quad\Longrightarrow\quad v_i\le F'_-(Y).
    \label{eq:zero-burn-executed-slope}
\end{equation}
Indeed, if $v_i>F'_-(Y)$, the user can unwind a sufficiently small amount
$\varepsilon$ of its trade.  It returns
$F(Y)-F(Y-\varepsilon)$ units of $X$ and recovers $\varepsilon$ units of
$Y$, increasing its utility by
\[
    v_i\varepsilon-\bigl(F(Y)-F(Y-\varepsilon)\bigr)>0.
\]
Exact accounting is preserved and, for small enough $\varepsilon$, the
user's net trade remains in its declared direction.  Conversely, an
underfilled user must satisfy
\begin{equation}
    y_i<q_i \quad\Longrightarrow\quad v_i\ge F'_+(Y),
    \label{eq:zero-burn-underfilled-slope}
\end{equation}
because otherwise it could sell a sufficiently small additional amount
directly to the pool and receive $F(Y+\varepsilon)-F(Y)$ units of $X$.

We now choose the witness types using a strict gap between the marginal
compensation at input $Q$ and the average compensation over the whole trade.
There exists $Q>0$ such that
\begin{equation*}
    \sigma:=F'_+(Q)<\frac{F(Q)}{Q}.
\end{equation*}
To see this, concavity and $F(0)=0$ always give
$F'_+(Q)\le F(Q)/Q$.  If equality held for every $Q>0$, $F$ would be
linear with its positive initial slope on $[0,+\infty)$, contradicting the
finite upper reserve bound $F(Q)<x_0$.  Moreover, $\sigma>0$: if the right
derivative of a concave increasing function vanished at a finite point, the
function would be constant thereafter, contrary to strict increasingness.

Choose numbers
\begin{equation}
    0<v^*<\sigma<v_1<v_2<\frac{F(Q)}{Q}.
    \label{eq:zero-burn-type-choice}
\end{equation}
All three asks are strictly eligible, since concavity gives
$F(Q)/Q\le F'_+(0)=\sigma_0$.  Consider two users, each with cap $Q$:
the first user's true valuation is $v_1$, and the second user's true
valuation is $v_2$.  Write $(x_1,y_1)$ and $(x_2,y_2)$ for their respective
compensations and allocations at the profile under consideration.
The second user's valuation exceeds the marginal compensation $\sigma$, but is
below the average compensation $F(Q)/Q$.  Thus supplying the whole amount
$Q$ would give it positive utility if it received all of $F(Q)$, even though
supplying additional units beyond $Q$ would not be worthwhile at the margin.
The ask $v^*<\sigma$
will be its dishonest report.

We first show that the second user sells zero under truthful reports:
the cheaper first user prevents it from obtaining any allocation.
If $y_2>0$, then
\eqref{eq:zero-burn-priority} forces $y_1=Q$.  Hence the total input satisfies
$Y>Q$, and \eqref{eq:zero-burn-slope-order} and
\eqref{eq:zero-burn-executed-slope} give
\[
    v_2\le F'_-(Y)\le F'_+(Q)=\sigma,
\]
contrary to $v_2>\sigma$.  Thus $y_2=0$, and
\eqref{eq:zero-burn-zero-allocation} gives $x_2=0$.  The second user's
truthful utility is therefore zero.

Now let the second user deviate only by reporting ask $v^*$, keeping its cap
equal to $Q$.  This reverses the priority: the second user now appears
cheaper and, as we show next, becomes the sole supplier.
The numerical profile $(v_1,v^*)$ can equally be viewed as a truthful profile,
so Pareto optimality applies to its outcome with respect to these reported
valuations.  If $y_1>0$, priority
\eqref{eq:zero-burn-priority} forces $y_2=Q$.  Then $Y>Q$, and the executed
first user would have to satisfy
\[
    v_1\le F'_-(Y)\le F'_+(Q)=\sigma,
\]
again a contradiction.  Hence $y_1=0$, and
\eqref{eq:zero-burn-zero-allocation} gives $x_1=0$.

Finally, the second user must sell its full cap under the deviation.  Otherwise
$Y=y_2<Q$, while concavity gives
\[
    F'_+(Y)\ge F'_+(Q)=\sigma>v^*,
\]
contradicting the underfilled condition
\eqref{eq:zero-burn-underfilled-slope}.  Therefore $y_2=Q$.
At this point, zero burn is decisive: the builder receives no fee, and the
inactive first user receives nothing, so the entire AMM output must go to
the second user.  Exact accounting
\eqref{eq:zero-burn-exact-accounting} forces
\[
    x_2=F(Q).
\]
Evaluated at its true valuation $v_2$, rather than its report $v^*$, the second
user's deviating utility is
\[
    F(Q)-v_2Q>0
\]
by \eqref{eq:zero-burn-type-choice}, whereas its truthful utility was zero.
This contradicts UIC.  All reports used in the argument have the same fixed
cap, and every ask used in a Pareto comparison is strictly eligible.  The
claimed impossibility follows.
\end{proof}

\begin{remark}
If builder fees are prohibited, then the impossibility in \Cref{thm:zero-burn-impossibility} holds without
assuming strategy-proofness for a builder-as-user.
Indeed, the proof uses builder-as-user strategy-proofness only to invoke
the zero-builder-fee theorem (\Cref{thm:zero-builder-fee}), which is unnecessary
when zero builder fees are imposed by assumption.  
\end{remark}

\end{document}